\documentclass[journal]{IEEEtran}

\ifCLASSINFOpdf

\else

\fi

\usepackage{mathrsfs}
\usepackage{amsmath}
\usepackage{amsfonts}
\usepackage{amsthm}
\usepackage{amssymb}
\usepackage{graphicx}
\usepackage{epstopdf}
\usepackage{indentfirst}
\usepackage{array}
\usepackage{cite}
\usepackage{enumerate}
\newtheorem{theorem}{\bf{Theorem}}

\newtheorem{remark}{\bf{Remark}}
\newtheorem{proposition}{\bf{Proposition}}

\newtheorem{definition}{\bf{Definition}}

\begin{document}

\title{Does Gaussian Approximation Work Well for The Long-Length Polar Code Construction?}

\author{Jincheng~Dai,~\IEEEmembership{Student Member,~IEEE},
        Kai~Niu,~\IEEEmembership{Member,~IEEE}, Zhongwei~Si,~\IEEEmembership{Member,~IEEE}, Chao~Dong,~\IEEEmembership{Member,~IEEE} and Jiaru Lin,~\IEEEmembership{Member,~IEEE}
\thanks{This work is supported by the National Natural Science Foundation of China (No. 61671080 \& No. 61171099), BUPT-SICE Excellent Graduate Students Innovation Fund and Huawei HIRP project. The material in this paper was presented in part at the Recent Results Session of IEEE International Symposium on Information Theory (ISIT), Hong Kong, June 2015.}
\thanks{The authors are with the Key Laboratory of Universal Wireless Communications, Ministry of Education, Beijing University of Posts and Telecommunications, Beijing 100876, China (email: \{daijincheng, niukai, sizhongwei, dongchao, jrlin\}@bupt.edu.cn).}
\vspace{-1em}
}

\maketitle

\begin{abstract}
Gaussian approximation (GA) is widely used to construct polar codes. However when the code length is long, the subchannel selection inaccuracy due to the calculation error of conventional approximate GA (AGA), which uses a two-segment approximation function, results in a catastrophic performance loss. In this paper, new principles to design the GA approximation functions for polar codes are proposed. First, we introduce the concepts of polarization violation set (PVS) and polarization reversal set (PRS) to explain the essential reasons that the conventional AGA scheme cannot work well for the long-length polar code construction. In fact, these two sets will lead to the rank error of subsequent subchannels, which means the orders of subchannels are misaligned, which is a severe problem for polar code construction. Second, we propose a new metric, named cumulative-logarithmic error (CLE), to quantitatively evaluate the remainder approximation error of AGA in logarithm. We derive the upper bound of CLE to simplify its calculation. Finally, guided by PVS, PRS and CLE bound analysis, we propose new construction rules based on a multi-segment approximation function, which obviously improve the calculation accuracy of AGA so as to ensure the excellent performance of polar codes especially for the long code lengths. Numerical and simulation results indicate that the proposed AGA schemes are critical to construct the high-performance polar codes.
\end{abstract}


\IEEEpeerreviewmaketitle

\section{Introduction}

\IEEEPARstart{P}{olar} codes proposed by Ar{\i}kan \cite{arikan} have been proved to achieve the capacity of any symmetric binary input symmetric discrete memoryless channels (B-DMCs) under a successive cancellation (SC) decoder as the code length goes to infinity. Recently, polar codes have been identified as one of the channel coding schemes in the 5G wireless communication system due to its excellent performance \cite{3GPP1}. To construct polar codes, the channel reliabilities are calculated efficiently using the symmetric capacities of subchannels or the Bhattacharyya parameters for the binary-input erasure channels (BECs). As a heuristic method, Ar{\i}kan has suggested to use the recursion which is optimal only for BECs also for other B-DMCs \cite{arikan_rm}. Mori \emph{et al.} regarded the construction problem as an instance of density evolution (DE) \cite{mori_de1}, which theoretically has the highest accuracy. Considering its high computational complexity, Tal and Vardy devised two approximation methods to simplify the calculation of DE, by which one can get the upper and lower bounds on the error probability of each subchannel. Tal and Vardy's method has almost no performance loss compared with DE \cite{talvardy,Pedarsani,talvardy_ISIT,talvardy_simplified}. Afterwards, Gaussian approximation (GA) was proposed to further reduce the computational complexity of DE \cite{trifonov} without much sacrifice in accuracy, which became popular in the construction of polar codes thanks to its good tradeoff between the complexity and performance.

In the GA construction of polar codes, the bit log-likelihood ratio (LLR) of each subchannel is assumed to obey a constraint Gaussian distribution in which the mean is half of the variance. Hence, the iterative evaluation of each subchannel reliability is only involved with the mean update of LLRs. However, the LLR mean updates in check nodes still depend on complex integration. Consequently, for construction of polar codes, the computational complexity of exact GA (denoted by EGA) grows exponentially with the polarization levels. This makes EGA too complicated to be practically employed. Therefore, in practical implementation, like GA utilized in LDPC codes, the well-known approximate version of GA (denoted by AGA) given by Chung \emph{et al.} based on a two-segment approximation function is used to speed up the calculations \cite{ldpc_2GA,xidian_GA,practical_GA}.

Initially, the approximation function in conventional AGA chosen by Chung is suitable for LDPC codes. However, in principle, we don't know whether this approximation method can be also good for the polar codes. In fact, the calculation error of AGA versus EGA will be accumulated and amplified in the recursion process of polar codes construction. Consequently, this phenomenon causes the inaccurate subchannel selection and results in a catastrophic block error ratio (BLER) performance loss for the long code lengths. Taking the polar code structure characteristics into consideration, it is lack of the comprehensive framework to design the AGA schemes for the construction of polar codes. In addition, the performance evaluation method for different AGA is also absent.

On the polar code construction, we find that the rank error of the subchannels and the calculation error of each subchannel's reliability are two critical factors to affect the accuracy of AGA approximation function. Here, ranking error means the orders of subchannels are misaligned. For various AGA schemes, the two factors will result in different evaluation error of each subchannel. Followed by above two distortions, we reveal the essential reason that the conventional AGA scheme leads to a catastrophic performance loss. Our ultimate goal is to propose systematic design rules of the AGA approximation function for polar codes, which achieves the excellent performance as well as reduces the GA computational complexity.

Our aim in this paper is to provide the new principles to design the multi-segment GA approximation functions so as to improve the calculation accuracy of AGA and guarantee the excellent performance of polar codes. The main contributions can be summarized in the following three aspects:
\begin{itemize}
  \item First, we take a closer investigation at the reason behind the poor performance of the long-length polar codes when the conventional versions of AGA (e.g. Chung's scheme) are used. To this end, we introduce the concepts of polarization violation set (PVS) and polarization reversal set (PRS). In the AGA process, when the subchannel's LLR mean belongs to the two sets, it will bring in the rank error and the polarization is `violated' or `reverted' among the subsequent subchannels. This phenomenon is not consistent with Ar{\i}kan's fundamental polarization relationship. The two sets reveal the essential reason that polar codes constructed by the conventional AGA present poor performance at long code lengths.
  \item Second, after eliminating PVS and PRS, we further propose a new metric, named cumulative-logarithmic error (CLE) of channel polarization, to quantitatively evaluate the remainder calculation error between AGA and EGA in the construction of polar codes. We also derived the upper bound of CLE to simplify its calculation. With this bound, the performance of different versions of AGA can be easily evaluated by analytic calculation rather than redundant the Monte-Carlo simulation.
  \item Finally, guided by PVS, PRS and the CLE bound, we propose new design rules for the improved AGA techniques which is tailored for the polar code construction. In this way, a systematic framework is established to design the high accuracy and low complexity AGA scheme for polar codes at any code length. Followed by the proposed rules, three new AGA schemes are given to guarantee the excellent performance of polar codes.
\end{itemize}

The remainder of the paper is organized as follows. The preliminaries of polar coding are described in Section II. Then the conventional GA is introduced in Section III. Section IV makes detailed error analysis of GA, in which the concepts of PVS, PRS and CLE are proposed. Then the new design rules of AGA approximation functions are given in Section V, and new AGA schemes with complexity comparison are also given in this part. Different versions of AGA are compared with the help of CLE bound in Section VI, where the simulation results are also analyzed. Finally, Section VII concludes this paper.

\section{Preliminaries}

\subsection{Notation Conventions}
In this paper, we use calligraphic characters, such as ${\cal X}$, to denote sets. Let $\left| \cal X \right|$ denote cardinality of $\cal X$. We write lowercase letters (e.g., $x$) to denote scalars. We use notation $v_1^N$ to denote a vector $\left( {{v_1},{v_2}, \cdots ,{v_N}} \right)$ and $v_i^j$ to denote a subvector $\left( {{v_i},{v_{i + 1}}, \cdots ,{v_j}} \right)$. The sets of binary and real field are denoted by $\mathbb{B}$ and $\mathbb{R}$, respectively. Specially, let ${\cal N}(a,b)$ denote Gaussian distribution, where $a$ and $b$ represent the mean and the variance respectively. For polar coding, only square matrices are involved in this paper, and they are denoted by bold letters. The subscript of a matrix indicates its size, e.g. ${{\bf{F}}_N}$ represents an $N \times N$ matrix ${\bf{F}}$. The Kronecker product of two matrices ${\bf{F}}$ and ${\bf{G}}$ is expressed as ${\bf{F}} \otimes {\bf{G}}$, and the $n$-fold Kronecker power of ${\bf{F}}$ is denoted by ${{\bf{F}}^{ \otimes n}}$.

Throughout this paper, $\log \left(  \cdot  \right)$ means ``logarithm to base 2'', and $\ln \left(  \cdot  \right)$ stands for the natural logarithm.

\subsection{Polar Codes and SC Decoding}
Let $W$ : $\cal X \to \cal Y$ denote a B-DMC with input alphabet $\cal X$ and output alphabet $\cal Y$. The channel transition probabilities are given by $W\left( {y\left| x \right.} \right)$, $x \in \cal X$ and $y \in {\cal Y}$. Given the code length $N = {2^n}$, $n = 1,2, \cdots $, the information length $K$, and the code rate $R = K/N$, the polar coding is described as \cite{arikan}. After the channel combining and splitting operations on $N$ independent duplicates of $W$, we obtain $N$ successive uses of synthesized binary input channels $W_N^{\left( j \right)}$, $j = 1,2, \cdots ,N$, with transition probabilities $W_N^{\left( j \right)}( {y_1^N,u_1^{j - 1}\left| {{u_j}} \right.} )$. The information bits can be assigned to the channels with indices in the information set ${\cal A}$, which are the more reliable subchannels. The complementary set ${\cal A}^c$ denotes the frozen bit set and the frozen bits ${u_{{{\cal A}^c}}}$ can be set as the fixed bit values, such as all zeros, for the symmetric channels. To put it in another way \cite{arikan}, polar coding is performed on the constraint $x_1^N = u_1^N{{\bf{G}}_N}$, where ${{\bf{G}}_N}$ is the generator matrix and $u_1^N,x_1^N \in {\left\{ {0,1} \right\}^N}$ are the source and code block respectively. The source block $u_1^N$ consists of information bits ${u_{{{\cal A}}}}$ and frozen bits ${u_{{{\cal A}^c}}}$. The generator matrix can be defined as ${{\bf{G}}_N} = {{\bf{B}}_N}{\bf{F}}_2^{ \otimes n}$, where ${{\bf{B}}_N}$ is the bit-reversal permutation matrix and ${{\bf{F}}_2} = \left[ { \begin{smallmatrix} 1 & 0 \\ 1 &  1 \end{smallmatrix} } \right]$.

As mentioned in \cite{arikan}, polar codes can be decoded by successive cancellation (SC) decoding algorithm. Let $\hat u_1^N$ denote an estimate of source block $u_1^N$. After receiving $y_1^N$, the bits $\hat u_j$ are successively determined with index from $1$ to $N$ in the following way:
\begin{equation}\label{SC_decision}
  {\hat u_j} = \left\{ \begin{array}{ll}
    {h_j}( {y_1^N,\hat u_1^{j - 1}} ) & {j \in \cal A},\\
    {u_j} & j \in {{\cal A}^c},
    \end{array} \right.
\end{equation}
where
\begin{equation}\label{SC_decision_where}
  {h_j}( {y_1^N,\hat u_1^{j - 1}} ) = \left\{ \begin{array}{ll}
    0 & {\rm{if}}~\frac{{W_N^{\left( j \right)}\left( {y_1^N,\hat u_1^{j - 1}\left| 0 \right.} \right)}}{{W_N^{\left( j \right)}\left( {y_1^N,\hat u_1^{j - 1}\left| 1 \right.} \right)}} \ge 1,\\
    1 & \rm{otherwise}.
    \end{array} \right.
\end{equation}
Given a polar code with code length $N$, information length $K$ and selected channels indices ${\cal A}$, the BLER under SC decoding algorithm is upper bounded by
\begin{equation}\label{SCbound}
  {P_e}\left( {N,K,{\cal A}} \right) \le \sum\limits_{j \in {\cal A}} {{P_e}\left( {W_N^{\left( j \right)}} \right)},
\end{equation}
where ${{P_e}( {W_N^{\left( j \right)}} )}$ is the error probability of the $j$-th subchannel. This BLER upper bound is named as the \emph{SC bound}.

\section{Gaussian Approximation for Polar Codes}
In this section, we use the code tree to describe the process of channel polarization. Based on the tree structure, we present and analyze the basic procedure of GA.

\subsection{Code Tree}
The channel polarization process can be expressed on a code tree. For a polar code with length $N=2^n$, the corresponding code tree $\cal T$ is a perfect binary tree\footnote{A perfect binary tree is a binary tree in which all interior nodes have two children and all leaves have the same depth or same level.}. Specifically, $\cal T$ can be represented as a $2$-tuple $\left( {\cal V,\cal B} \right)$, where $\cal V$ and $\cal B$ denote the set of nodes and the set of edges, respectively.

Depth of a node is the length of the path from the root to this node. The set of all the nodes at a given depth $i$ is denoted by ${{\cal V}_i}$, $i = 0,1,2, \cdots ,n$. The root node has a depth of zero. Let $v_i^{\left( j \right)}$, $j = 1,2, \cdots ,{2^i}$, denote the $j$-th node from left to right in ${{\cal V}_i}$. As an illustration, Fig. \ref{polar_codetree} shows a toy example of code tree with $N = 16$, which includes $4$ levels. In the nodes set ${{\cal V}_2}$, the $2$nd node from left to right is denoted by $v_2^{\left( 2 \right)}$. Except for the nodes at the $n$-th depth, each $v_i^{\left( j \right)} \in {{\cal V}_i}$ has two descendants in ${{\cal V}_{i+1}}$, and the two corresponding edges are labeled as 0 and 1, respectively. The nodes $v_n^{\left( j \right)} \in {\cal V}_n$ are called leaf nodes. Let ${\cal T}( {v_i^{\left( j \right)}} )$ denote a subtree with a root node ${v_i^{\left( j \right)}}$. The depth of this subtree can be defined as $n-i$ which indicates the difference between the depth of the leaf node and that of the root node. In addition, the node ${v_i^{\left( j \right)}}$ has two subtrees, that is, the left subtree ${{\cal T}_{{\rm{left}}}} = {\cal T}(v_{i + 1}^{\left( {2j - 1} \right)})$ and the right subtree ${{\cal T}_{{\rm{right}}}} = {\cal T}(v_{i + 1}^{\left( {2j} \right)})$.

All the edges in the set $\cal B$ are partitioned into $n$ levels ${\cal B}_l$, $l = 1,2, \cdots ,n$. Each edge in the $l$-th level ${\cal B}_l$ is incident to two nodes: one at depth $l-1$ and the other at depth $l$. An $i$-depth node is corresponding to a path $\left( {{b_1},{b_2}, \cdots ,{b_i}} \right)$ which consists of $i$ edges, with $b_l \in {\cal B}_l$, $l = 1,2, \cdots ,i$. A vector $b_1^i = \left( {{b_1},{b_2}, \cdots ,{b_i}} \right)$ is used to depict the above path.

\subsection{Gaussian Approximation for Polar Codes}
Trifonov \cite{trifonov} suggests a polar code construction method for the binary input AWGN (BI-AWGN) channels based on a Gaussian assumption in every recursion step. For the BI-AWGN channels with noise variance ${\sigma ^2}$, the coded bits are modulated using binary phase shift keying (BPSK). The transition probability $W(y|x)$ is written as
\begin{equation}\label{BIAWGN_trans}
  W\left( {y\left| x \right.} \right) = \frac{1}{{\sqrt {2\pi {\sigma ^2}} }}{{e}^{ - \frac{{{{\left( {y - \left( {1 - 2x} \right)} \right)}^2}}}{{2{\sigma ^2}}}}},
\end{equation}
where $x \in \mathbb{B}$ and $y \in \mathbb{R}$. The LLR of each received symbol $y$ is denoted by
\begin{equation}\label{BIAWGN_LLR}
  L\left( y \right) = \ln \frac{{W\left( {y\left| 0 \right.} \right)}}{{W\left( {y\left| 1 \right.} \right)}} = \frac{{2y}}{{{\sigma ^2}}}.
\end{equation}
Without loss of generality, we assume that all-zero codeword is transmitted. One can check $L\left( {{y}} \right) \sim {\cal N}\left( {\frac{2}{{{\sigma ^2}}},\frac{4}{{{\sigma ^2}}}} \right)$.

{\bf{GA assumption:}} The LLR of each subchannel obeys a constraint Gaussian distribution in which the mean is half of the variance \cite{trifonov,ldpc_2GA,xidian_GA,practical_GA}.

According to the GA assumption, the only issue needed to be dealt with is the LLR mean. Therefore, in the construction of polar codes, to obtain the reliability of each subchannel, we trace their LLR mean. This recursive calculation process can be performed on the code tree. The set of LLRs corresponding to the nodes at depth $i$ is denoted by ${{\cal L}_i}$, $i = 0,1,2, \cdots ,n$. Let $L_i^{\left( j \right)}$, $j = 1,2, \cdots ,{2^i}$, denote the $j$-th element in ${{\cal L}_i}$. We write $m_i^{\left( j \right)}$ as the mean of $L_i^{\left( j \right)}$. So the mean of LLR from the channel information can be written as $m_0^{\left( 1 \right)} = \frac{2}{{{\sigma ^2}}}$, and under the GA assumption we have
\begin{equation}\label{LLR_dist}
  L_i^{\left( j \right)} \sim {\cal N}\left( {m_i^{\left( j \right)},2m_i^{\left( j \right)}} \right).
\end{equation}
Here, $m_i^{\left( j \right)}$ can be computed recursively as
\begin{equation}
  m_{i + 1}^{\left( {2j - 1} \right)} = {f_c}\left( {m_i^{\left( j \right)}} \right),m_{i + 1}^{\left( {2j} \right)} = {f_v}\left( {m_i^{\left( j \right)}} \right),
\end{equation}
where the functions ${f_c}\left( t \right)$ and ${f_v}\left( t \right)$ are used for check nodes (left branch) and variable nodes (right branch), respectively. The physical meaning of the function variable $t$ stands for subchannel's LLR mean in GA construction. We have
\begin{equation}\label{f1f2}
    \left\{ \begin{array}{l}
    {f_c}\left( t \right) = {\phi ^{ - 1}}\left( {1 - {{\left( {1 - \phi \left( t \right)} \right)}^2}} \right),\\
    {f_v}\left( t \right) = 2t.
    \end{array} \right.
\end{equation}
In EGA, $\phi \left( t \right)$ is written as
\begin{equation}
    {\phi \left( t \right) = \left\{ \begin{array}{ll}
    1 - \frac{1}{{\sqrt {4\pi t} }}\int_{\mathbb{R}} {\tanh \left( \frac{z}{2} \right){{e}^{ - \frac{{{{\left( {z - t} \right)}^2}}}{{4t}}}}dz} & t > 0,\\
    1 & t = 0,
    \end{array} \right.}
\end{equation}
where $\tanh \left(  \cdot  \right)$ denotes hyperbolic tangent function. It is easy to check that $\phi \left( t \right)$ is continuous and monotonically decreasing on $\left[ {0, + \infty } \right)$, with $\phi \left( 0 \right) = 1$ and $\phi \left( { + \infty } \right) = 0$ \cite{ldpc_2GA}. As an illustration, the red bold edge in Fig. \ref{polar_codetree} depicts the recursive calculation process of $m_4^{\left( 8 \right)}$, whose corresponding path is denoted by $\left( {{b_1},{b_2},{b_3},{b_4}} \right) = \left( {0,1,1,1} \right)$.

\begin{figure}[t]
\vspace{0.3em}
\setlength{\abovecaptionskip}{0.cm}
\setlength{\belowcaptionskip}{-0.cm}
  \centering{\includegraphics[scale=1.25]{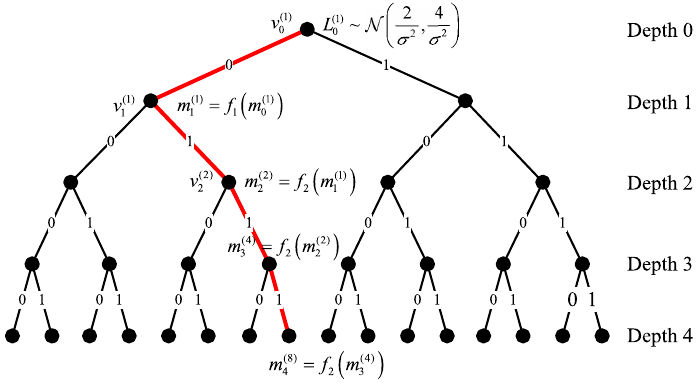}}
  \caption{An example of code tree for $N = 16$, $n = 4$. The red bold edge shows the recursive calculation process of $m_4^{\left( 8 \right)}$.}\label{polar_codetree}
  \vspace{-1em}
\end{figure}

Obviously, the exact calculation of LLR mean in check nodes requires complex integration, which results in a high computational complexity. Therefore, Chung \emph{et al.} give the well known two-segment approximation function of $\phi \left( t \right)$, denoted by $\varphi\left( t \right)$, for the analysis of LDPC codes in \cite{ldpc_2GA},
\begin{equation}\label{convGA}
    {\varphi\left( t \right) = \left\{ \begin{array}{ll}
    {{e}^{ - 0.4527{t^{0.86}} + 0.0218}} & 0 < t < 10,\\
    \sqrt {\frac{\pi }{t}} {{e}^{ - \frac{t}{4}}} \left( {1 - \frac{{10}}{{7t}}} \right) & t \ge 10.
    \end{array} \right.}
\end{equation}
(\ref{convGA}) is also widely used in the construction of polar codes \cite{xidian_GA}. Its corresponding AGA algorithm is denoted by ``Chung".

By the GA assumption of (\ref{LLR_dist}), the error probabilities of polarized subchannel ${P_e}( {W_N^{\left( j \right)}} )$ can be written as
\begin{equation}
    {P_e}\left( {W_N^{\left( j \right)}} \right) = Q\left( {\frac{{m_n^{\left( j \right)}}}{{\sqrt {2m_n^{\left( j \right)}} }}} \right) = Q\left( {\sqrt {\frac{{m_n^{\left( j \right)}}}{2}} } \right),
\end{equation}
where $Q\left( \varsigma \right) = \frac{1}{{\sqrt {2\pi } }}\int_{\varsigma}^{ + \infty } {{{e}^{ - \frac{{{z^2}}}{2}}}dz}$. Thus, the SC bound can be written as
\begin{equation}
  {P_e}\left( {N,K,\cal{A}} \right) \le \sum\limits_{j \in \cal{A}} {Q\left( {\sqrt {\frac{{m_n^{\left( j \right)}}}{2}} } \right)}.
\end{equation}
Since $Q\left( \varsigma \right)$ is a monotone decreasing function, the subchannel $W_N^{\left( j \right)}$ with a larger mean $m_n^{\left( j \right)}$ has higher reliability. The construction of polar codes corresponds to the selection of best $K$ subchannels among $N$ as information set ${\cal A}$ in terms of the LLR means $m_n^{\left( j \right)}$, where $j = 1,2,\cdots,N$.

\section{Error Analysis of Gaussian Approximation}
In this section, we introduce the concepts of polarization violation set (PVS) and polarization reversal set (PRS). By calculating the two sets, we demonstrate the intrinsic reason that polar codes constructed by conventional AGA suffer from catastrophic performance loss at long code lengths. In order to quantitatively evaluate the remainder calculation error between AGA and EGA, we further propose the concept of cumulative-logarithmic error (CLE) of channel polarization and give a bound to simplify its calculation. Based on the CLE bound, we can efficiently evaluate the performance of different approximation functions in AGA.

\subsection{PVS and PRS}
\begin{proposition}\label{proposition_capacity_LLRmean}
  \emph{Under the GA assumption, each subchannel's symmetric capacity monotonically increases with its LLR mean.}
\end{proposition}

It is worth noticing that a BI-AWGN channel's capacity $I\left( W \right)$ is written as
\begin{equation}\label{BI_AWGN}
\begin{aligned}
I\left( W \right) & = h\left( {{\sigma ^2}} \right)\\
~ & \buildrel \Delta \over = \frac{1}{2}\sum\limits_{x \in \mathbb{B}} {\int_{\mathbb{R}} {W\left( {y\left| x \right.} \right)\log \left( {\frac{{2W\left( {y\left| x \right.} \right)}}{{W\left( {y\left| 0 \right.} \right) + W\left( {y\left| 1 \right.} \right)}}} \right)dy} },
\end{aligned}
\end{equation}
where the transition probability $W\left( {y\left| x \right.} \right)$ is given as (\ref{BIAWGN_trans}). In terms of the GA principle, each subchannel is approximated by a BI-AWGN channel $W$ with LLR mean $m$. In addition, the variance of corresponding additive white Gaussian noise is ${\sigma ^2} = \frac{2}{m}$ under GA assumption. Since the function $h\left( {{\sigma ^2}} \right)$ monotonically decreases with ${\sigma ^2}$, the symmetric capacity $I\left( W \right)$ monotonically increases with its LLR mean $m$. In addition, we have
\begin{equation}\label{capacity_limits}
\mathop {\lim }\limits_{m \to 0} I\left( W \right) = 0,\mathop {\lim }\limits_{m \to  + \infty } I\left( W \right) = 1.
\end{equation}

\begin{proposition}\label{Proposition_LLRmean_relationship}
  \emph{For the size-two channel polarization, suppose $\left( {W,W} \right) \mapsto ( {{W_2^{(1)} },{W_2^{(2)} }} )$. Under the GA assumption, the LLR means corresponding to $W$, ${W_2^{(1)} }$ and ${W_2^{(2)} }$ are represented as $m$, ${m_2^{(1)} }$ and ${m_2^{(2)} }$, respectively. Then, the LLR means should satisfy
    \begin{equation}\label{LLR_inequlity}
      {m_2^{(1)} } \le m \le {m_2^{(2)} }
    \end{equation}
  with equality if and only if $m = 0$ or $m = +\infty$.
  }
\end{proposition}
This result follows the \cite[Proposition 4]{arikan}. Combining with \emph{Proposition \ref{proposition_capacity_LLRmean}}, the proof of (\ref{LLR_inequlity}) is immediate. It can be seen from (\ref{LLR_inequlity}) that the reliability of the original channel $W$ is redistributed. Based on this interpretation, we may say that after one step polarization, a ``bad'' channel ${W_2^{(1)}}$ and a ``good'' channel ${W_2^{(2)}}$ have been created.

\begin{proposition}\label{proposition_omega}
   \emph{In the AGA construction of polar codes, the approximation function $\Omega \left( t \right)$ of $\phi \left( t \right)$ should satisfy
\begin{equation}\label{Theorem_1_Omega_satisfy}
  0 < \Omega \left( t \right) < 1,
\end{equation}
which is the necessary and sufficient condition for $\Omega \left( t \right)$ to satisfy \emph{Proposition 2}.
}
\end{proposition}

\begin{proof}
    Different from that in LDPC codes, the approximation function $\Omega \left( t \right)$ for polar codes should guarantee that the relationship of (\ref{LLR_inequlity}) holds\footnote{In this paper, we analyze the error between AGA and EGA, rather than the error of GA itself.}. Therefore, in the size-two channel polarization, for $t \in \left( {0, + \infty } \right)$, $\Omega \left( t \right)$ should satisfy
    \begin{equation}\label{omega_satisfy}
       {{\Omega }^{ - 1}}\left( {1 - {{\left( {1 - \Omega \left( t \right)} \right)}^2}} \right) < t < 2t,
    \end{equation}
    which follows \emph{Proposition \ref{Proposition_LLRmean_relationship}}. Recall that $\phi \left( t \right)$ is continuous and monotonically decreasing on $\left[ {0, + \infty } \right)$ \cite{ldpc_2GA}, we therefore assume its approximated form $\Omega \left( t \right)$ monotonically decreases on $\left( {0, + \infty } \right)$. Consequently, the left inequality in (\ref{omega_satisfy}) can be simplified as
    \begin{equation}\label{omega_satisfy_simplified}
        1 - {\left( {1 - \Omega \left( t \right)} \right)^2} > \Omega \left( t \right) \Rightarrow 0 < \Omega \left( t \right) < 1,
    \end{equation}
    In turn, if $\Omega \left( t \right)$ satisfies $0 < \Omega \left( t \right) < 1$, we have
    \begin{equation}
      \begin{aligned}
        ~ & 1 - {\left( {1 - \Omega \left( t \right)} \right)^2} > 1 - \left( {1 - \Omega \left( t \right)} \right)\\
         \Rightarrow & {\Omega ^{ - 1}}\left( {1 - {{\left( {1 - \Omega \left( t \right)} \right)}^2}} \right) < \underbrace {{\Omega ^{ - 1}}\left( {1 - \left( {1 - \Omega \left( t \right)} \right)} \right)}_{ = t}\\
         \Rightarrow & {\Omega ^{ - 1}}\left( {1 - {{\left( {1 - \Omega \left( t \right)} \right)}^2}} \right) < t < 2t
        \end{aligned}
    \end{equation}
    The above analysis indicates that (\ref{Theorem_1_Omega_satisfy}) is the necessary and sufficient condition for $\Omega \left( t \right)$ to satisfy \emph{Proposition 2}.
\end{proof}

If $\Omega \left( t \right)$ cannot meet (\ref{Theorem_1_Omega_satisfy}), its approximation error with respect to the exact $\phi \left( t \right)$ will result in the following two types of reliability rank error:
\begin{enumerate}[Type 1)]
  \item In the size-two channel polarization, if $\Omega \left( t \right)$ leads to $m \le {m_2^{(1)} } < {m_2^{(2)} }$, this error indicates that the reliabilities of subchannels are partially violated, which is named as the ``polarization violation'' phenomenon.
  \item Furthermore, when $\Omega \left( t \right)$ leads to $m < {m_2^{(2)} } \le {m_2^{(1)} }$. This error indicates the reliabilities of subchannels have been wrongly reversed, which is named as the ``polarization reversal'' phenomenon.
\end{enumerate}

\begin{definition}\label{Definition_PVS}
\emph{Given the approximation function $\Omega \left( t \right)$, the polarization violation set (PVS) ${{\cal S}_{{\rm{PVS}}}}$ is defined as
\begin{equation}\label{Def1}
  {{\cal S}_{{\rm{PVS}}}} = \left\{ {t\left| {t \le {{\Omega }^{ - 1}}\left( {1 - {{\left( {1 - \Omega \left( t \right)} \right)}^2}} \right) < 2t} \right.} \right\},
\end{equation}
where $t \in \left( {0, + \infty } \right)$.
}
\end{definition}

Obviously, in the size-two channel polarization, for any LLR mean $m$ belonging to ${{\cal S}_{{\rm{PVS}}}}$, $\Omega \left( t \right)$ will certainly lead to $m \le {m_2^{(1)} } < {m_2^{(2)}}$, which violates the basic order in \emph{Proposition \ref{Proposition_LLRmean_relationship}}. Therefore, for the AGA algorithm with $\Omega \left( t \right)$, if ${{\cal S}_{{\rm{PVS}}}} \ne \varnothing$, any subchannel whose LLR mean belongs to ${{\cal S}_{{\rm{PVS}}}}$ will inaccurately create two ``good'' channel in the size-two channel polarization, which will lead to obvious approximation error in the subsequent AGA process.

\begin{definition}\label{Definition_PRS}
\emph{Given the approximation function $\Omega \left( t \right)$, the polarization reversal set (PRS) ${{\cal S}_{{\rm{PRS}}}}$ is defined as
\begin{equation}\label{Def2}
  {{\cal S}_{{\rm{PRS}}}} = \left\{ {t\left| {{{\Omega }^{ - 1}}\left( {1 - {{\left( {1 - \Omega \left( t \right)} \right)}^2}} \right) \ge 2t} \right.} \right\},
\end{equation}
where $t \in \left( {0, + \infty } \right)$.
}
\end{definition}

Interestingly, in the size-two channel polarization, for any LLR mean $m$ belonging to ${{\cal S}_{{\rm{PRS}}}}$, $\Omega \left( t \right)$ will result in $m < {m_2^{(2)} } \le {m_2^{(1)} }$. In other words, the split ``good" channel and ``bad'' channel swap their roles due to the calculation error of $\Omega \left( t \right)$, which yields severe error in the size-two polarization. This phenomenon then leads to the substantial error in subchannels' position rank\footnote{One should notice that we cannot have ${m_2^{(2)}} < m$ since ${m_2^{(2)}} = 2m$. There just exist three orders, which correspond to (\ref{LLR_inequlity}), PVS and PRS.}.

\begin{proposition}\label{proposition_PVS_PRS_relationship}
\emph{The relationship between PVS and PRS can be expressed as
\begin{equation}\label{PIS_PRS_Proposition}
  {{\cal S}_{{\rm{PRS}}}} \ne \varnothing \Rightarrow {{\cal S}_{{\rm{PVS}}}} \ne \varnothing.
\end{equation}
}
\end{proposition}

\begin{proof}
     Note that when $t \in \left( {0, + \infty } \right)$, for any $t \in {{\cal S}_{{\rm{PRS}}}}$, it makes ${{\Omega }^{ - 1}}( {1 - {{\left( {1 - \Omega \left( t \right)} \right)}^2}} ) \ge 2t$ by \emph{Definition \ref{Definition_PRS}}, which will inevitably result in ${{\Omega }^{ - 1}}( {1 - {{\left( {1 - \Omega \left( t \right)} \right)}^2}} ) \ge t$. Therefore, according to \emph{Definition \ref{Definition_PVS}}, the proof of \emph{Proposition \ref{proposition_PVS_PRS_relationship}} is immediate. In other words, ${{\cal S}_{{\rm{PRS}}}} \ne \varnothing$ is the sufficient condition of ${{\cal S}_{{\rm{PVS}}}} \ne \varnothing$. On the contrary, if ${{\cal S}_{{\rm{PVS}}}} = \varnothing$, we can derive ${{\cal S}_{{\rm{PRS}}}} = \varnothing$.
\end{proof}

Suppose $\Omega \left( t \right)$ monotonically decreases on $\left( {0, + \infty } \right)$, the left inequality in (\ref{Def1}) can therefore be simplified as
\begin{equation}\label{Def1_aid}
  \Omega \left( t \right) \ge 1 - {\left( {1 - \Omega \left( t \right)} \right)^2} \Rightarrow \Omega \left( t \right) \ge 1.
\end{equation}
Analogously, the inequality in (\ref{Def2}) will also be simplified as
\begin{equation}\label{Def2_aid}
  1 - {\left( {1 - \Omega \left( t \right)} \right)^2} \le \Omega \left( {2t} \right) \Rightarrow 2\Omega \left( t \right) - \Omega {\left( t \right)^2} \le \Omega \left( {2t} \right).
\end{equation}

\emph{Example:} recall that in Chung's conventional AGA scheme, the two-segment approximation function ${\varphi}\left( t \right)$ is a specific form of $\Omega \left( t \right)$. Since $\varphi \left( 0 \right) = {e^{0.0218}} > 1$, ${\varphi}\left( t \right)$ cannot satisfy (\ref{Theorem_1_Omega_satisfy}). For ${\varphi}\left( t \right)$, its corresponding PVS and PRS are denoted by ${{\cal S}_{{\rm{PVS}}}} = \left( {{a_1},{a_2}} \right]$ and ${{\cal S}_{{\rm{PRS}}}} = \left( {{0},{a_1}} \right]$, respectively. The boundary points $a_1$ and $a_2$ are given in the following equations
\begin{equation}\label{boun}
  \left\{ \begin{array}{l}
    2\varphi \left( {{a_1}} \right) - \varphi {\left( {{a_1}} \right)^2} = \varphi \left( {2{a_1}} \right),\\
    \varphi \left( {{a_2}} \right) = 1,
    \end{array} \right. \Rightarrow \left\{ \begin{array}{l}
    {a_1} = 0.01476,\\
    {a_2} = 0.02939,
    \end{array} \right.
\end{equation}
which follows (\ref{Def1_aid}) and (\ref{Def2_aid}). Hence for ${\varphi}\left( t \right)$ we have
\begin{equation}\label{AGA_2_PRS}
  {{\cal S}_{{\rm{PVS}}}} = \left( {0.01476,0.02939} \right],{{\cal S}_{{\rm{PRS}}}} = \left( {0,0.01476} \right],
\end{equation}
which are denoted in Fig. \ref{phi_func}.

\begin{figure}[t]
\vspace{0.4em}
\setlength{\abovecaptionskip}{0.cm}
\setlength{\belowcaptionskip}{-0.cm}
  \centering{\includegraphics[scale=0.74]{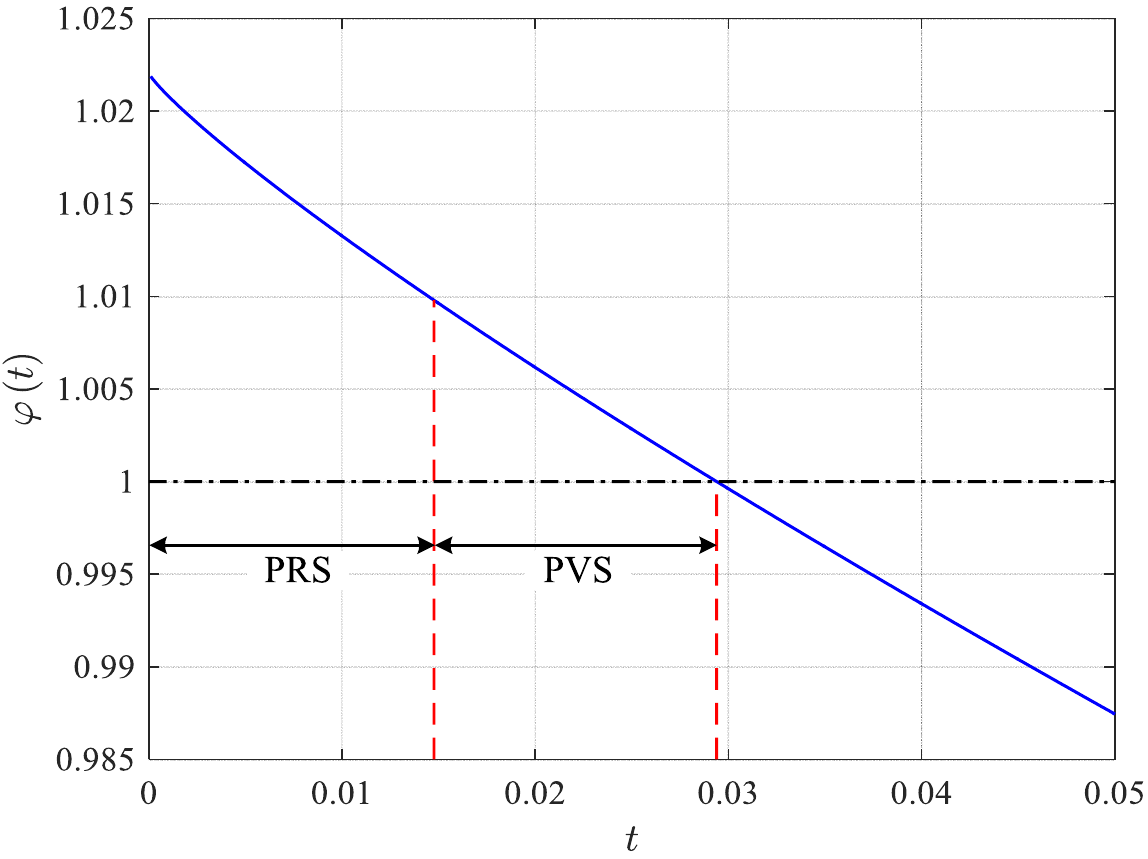}}
  \caption{Schematic plot of ${\varphi}\left( t \right)$, where ${{\cal S}_{{\rm{PVS}}}} = \left( {0.01476,0.02939} \right]$ and ${{\cal S}_{{\rm{PRS}}}} = \left( {0,0.01476} \right]$.}\label{phi_func}
  \vspace{-0.5em}
\end{figure}

\begin{theorem}\label{Theorem_2}
\emph{For the $N$-channel transform, where $N = 2^n$, $n \ge 1$, suppose that the original channel's LLR mean has two configurations, which are denoted by $\mathord{\buildrel{\lower3pt\hbox{$\scriptscriptstyle\frown$}}
\over m} _0^{\left( 1 \right)}$ and $\mathord{\buildrel{\lower3pt\hbox{$\scriptscriptstyle\smile$}}
\over m} _0^{\left( 1 \right)}$. If they satisfy $\mathord{\buildrel{\lower3pt\hbox{$\scriptscriptstyle\frown$}}
\over m} _0^{\left( 1 \right)} \ge \mathord{\buildrel{\lower3pt\hbox{$\scriptscriptstyle\smile$}}
\over m} _0^{\left( 1 \right)}$, then under the GA assumption, for $j = 1,2, \cdots ,{2^n}$, we have
\begin{equation}\label{Theorem1}
  \mathord{\buildrel{\lower3pt\hbox{$\scriptscriptstyle\frown$}}
    \over m} _n^{\left( j \right)} \ge \mathord{\buildrel{\lower3pt\hbox{$\scriptscriptstyle\smile$}}
    \over m} _n^{\left( j \right)}.
\end{equation}
}
\end{theorem}

\begin{proof}
  This result will be proved by mathematical induction. Under the GA assumption, recall that $\phi \left( t \right)$ is continuous and monotonically decreasing on $\left[ {0, + \infty } \right)$ \cite{ldpc_2GA}, one can easily check that ${f_c}\left( t \right)$ and ${f_v}\left( t \right)$ in (\ref{f1f2}) monotonously increase on $\left[ {0, + \infty } \right)$.

  Suppose $N = 2^k$, $k = 1$, if the two LLR mean configurations of the original channel satisfy $\mathord{\buildrel{\lower3pt\hbox{$\scriptscriptstyle\frown$}}
\over m} _0^{\left( 1 \right)} \ge \mathord{\buildrel{\lower3pt\hbox{$\scriptscriptstyle\smile$}}
\over m} _0^{\left( 1 \right)}$, then
\begin{equation}
  \left\{ \begin{array}{l}
    {f_c}\left( {\mathord{\buildrel{\lower3pt\hbox{$\scriptscriptstyle\frown$}}
    \over m} _0^{\left( 1 \right)}} \right) \ge {f_c}\left( {\mathord{\buildrel{\lower3pt\hbox{$\scriptscriptstyle\smile$}}
    \over m} _0^{\left( 1 \right)}} \right),\\
    {f_v}\left( {\mathord{\buildrel{\lower3pt\hbox{$\scriptscriptstyle\frown$}}
    \over m} _0^{\left( 1 \right)}} \right) \ge {f_v}\left( {\mathord{\buildrel{\lower3pt\hbox{$\scriptscriptstyle\smile$}}
    \over m} _0^{\left( 1 \right)}} \right),
    \end{array} \right. \Rightarrow \left\{ \begin{array}{l}
    \mathord{\buildrel{\lower3pt\hbox{$\scriptscriptstyle\frown$}}
    \over m} _1^{\left( 1 \right)} \ge \mathord{\buildrel{\lower3pt\hbox{$\scriptscriptstyle\smile$}}
    \over m} _1^{\left( 1 \right)},\\
    \mathord{\buildrel{\lower3pt\hbox{$\scriptscriptstyle\frown$}}
    \over m} _1^{\left( 2 \right)} \ge \mathord{\buildrel{\lower3pt\hbox{$\scriptscriptstyle\smile$}}
    \over m} _1^{\left( 2 \right)}.
    \end{array} \right.
\end{equation}
With the increased LLR mean of the original channel $W$, it can be seen that the LLR means of two polarized subchannels will strictly increase.

Next, suppose $N = 2^k$, if we have $\mathord{\buildrel{\lower3pt\hbox{$\scriptscriptstyle\frown$}}
\over m} _0^{\left( 1 \right)} \ge \mathord{\buildrel{\lower3pt\hbox{$\scriptscriptstyle\smile$}}
\over m} _0^{\left( 1 \right)}$, then for $j = 1,2, \cdots ,{2^k}$, $\mathord{\buildrel{\lower3pt\hbox{$\scriptscriptstyle\frown$}}
    \over m} _k^{\left( j \right)} \ge \mathord{\buildrel{\lower3pt\hbox{$\scriptscriptstyle\smile$}}
    \over m} _k^{\left( j \right)}$ holds. Thus, when $N = 2^{k+1}$, one can check
\begin{equation}
  \left\{ \begin{array}{l}
    {f_c}\left( {\mathord{\buildrel{\lower3pt\hbox{$\scriptscriptstyle\frown$}}
    \over m} _k^{\left( j \right)}} \right) \ge {f_c}\left( {\mathord{\buildrel{\lower3pt\hbox{$\scriptscriptstyle\smile$}}
    \over m} _k^{\left( j \right)}} \right),\\
    {f_v}\left( {\mathord{\buildrel{\lower3pt\hbox{$\scriptscriptstyle\frown$}}
    \over m} _k^{\left( j \right)}} \right) \ge {f_v}\left( {\mathord{\buildrel{\lower3pt\hbox{$\scriptscriptstyle\smile$}}
    \over m} _k^{\left( j \right)}} \right),
    \end{array} \right. \Rightarrow \left\{ \begin{array}{l}
    \mathord{\buildrel{\lower3pt\hbox{$\scriptscriptstyle\frown$}}
    \over m} _{k + 1}^{\left( {2j - 1} \right)} \ge \mathord{\buildrel{\lower3pt\hbox{$\scriptscriptstyle\smile$}}
    \over m} _{k + 1}^{\left( {2j - 1} \right)},\\
    \mathord{\buildrel{\lower3pt\hbox{$\scriptscriptstyle\frown$}}
    \over m} _{k + 1}^{\left( {2j} \right)} \ge \mathord{\buildrel{\lower3pt\hbox{$\scriptscriptstyle\smile$}}
    \over m} _{k + 1}^{\left( {2j} \right)}.
    \end{array} \right.
\end{equation}
In other words, for $j = 1,2, \cdots ,{2^{k+1}}$, $\mathord{\buildrel{\lower3pt\hbox{$\scriptscriptstyle\frown$}}
    \over m} _{k+1}^{\left( j \right)} \ge \mathord{\buildrel{\lower3pt\hbox{$\scriptscriptstyle\smile$}}
    \over m} _{k+1}^{\left( j \right)}$. From above analysis, the proof of (\ref{Theorem1}) is finished.
\end{proof}

\emph{Theorem \ref{Theorem_2}} indicates that under the GA assumption, if the LLR mean of the original channel $W$ increases, the polarized subchannel's LLR mean will increase together. Combining PVS and PRS, \emph{Theorem \ref{Theorem_2}} serves to analyze the subchannel's rank error and the poor performance of the long-length polar codes when Chung's conventional AGA is used.

For the AGA construction of polar codes with $\Omega \left( t \right)$, suppose ${{\cal S}_{{\rm{PRS}}}} \ne \varnothing$, then for any LLR mean $m_i^{\left( j \right)} \in {{\cal S}_{{\rm{PRS}}}}$, we have $m_{i + 1}^{\left( {2j - 1} \right)} \ge m_{i + 1}^{\left( {2j} \right)}$ by the definition of ${{\cal S}_{{\rm{PRS}}}}$ in the recursive calculation of GA. However, $\phi \left( t \right)$ in EGA can guarantee ${{\cal S}_{{\rm{PRS}}}} = \varnothing$ and ${{\cal S}_{{\rm{PVS}}}} = \varnothing$, which claims $m_{i + 1}^{\left( {2j - 1} \right)} < m_{i + 1}^{\left( {2j} \right)}$ for any $m_i^{\left( j \right)} > 0$. The above analysis indicates that in the AGA process, as presented in the code tree, due to ${{\cal S}_{{\rm{PRS}}}} \ne \varnothing$, the condition $m_i^{\left( j \right)} \in {{\cal S}_{{\rm{PRS}}}}$ will lead to the rank error among the leaf nodes which belong to the left subtree ${{\cal T}_{{\rm{left}}}} = {\cal T}(v_{i + 1}^{( {2j - 1} )})$ and the right subtree ${{\cal T}_{{\rm{right}}}} = {\cal T}(v_{i + 1}^{( {2j} )})$. In other words, we have
\begin{equation}\label{leaf_LLRmean}
  m_n^{\left( {\left( {j - 1} \right){2^{n - i}} + s} \right)} \ge m_n^{\left( {\left( {j - 1} \right){2^{n - i}} + s + {2^{n - i - 1}}} \right)},
\end{equation}
where $s = 1,2, \cdots ,{2^{n - i - 1}}$. This result follows from \emph{Theorem \ref{Theorem_2}} by mechanically applying $m_{i + 1}^{\left( {2j - 1} \right)} \Rightarrow \mathord{\buildrel{\lower3pt\hbox{$\scriptscriptstyle\frown$}}
\over m} _0^{\left( 1 \right)}$ and $m_{i + 1}^{\left( {2j} \right)} \Rightarrow \mathord{\buildrel{\lower3pt\hbox{$\scriptscriptstyle\smile$}}
\over m} _0^{\left( 1 \right)}$. Then we have $\mathord{\buildrel{\lower3pt\hbox{$\scriptscriptstyle\frown$}}
\over m} _{n - i - 1}^{\left( s \right)} \ge \mathord{\buildrel{\lower3pt\hbox{$\scriptscriptstyle\smile$}}
\over m} _{n - i - 1}^{\left( s \right)}$. Note that in these two subtrees, $\mathord{\buildrel{\lower3pt\hbox{$\scriptscriptstyle\frown$}}
\over m} _{n - i - 1}^{\left( s \right)}$ and $\mathord{\buildrel{\lower3pt\hbox{$\scriptscriptstyle\smile$}}
\over m} _{n - i - 1}^{\left( s \right)}$ correspond to the left and right side in (\ref{leaf_LLRmean}), respectively. However, in EGA, the ``$\ge$'' in (\ref{leaf_LLRmean}) should be ``$<$''. This polarization reversal phenomenon leads to the rank error of subchannel's position, which directly affects the selection of information set $\cal A$.

\begin{figure}[t]
\vspace{0.3em}
\setlength{\abovecaptionskip}{0.cm}
\setlength{\belowcaptionskip}{-0.cm}
  \centering{\includegraphics[scale=0.423]{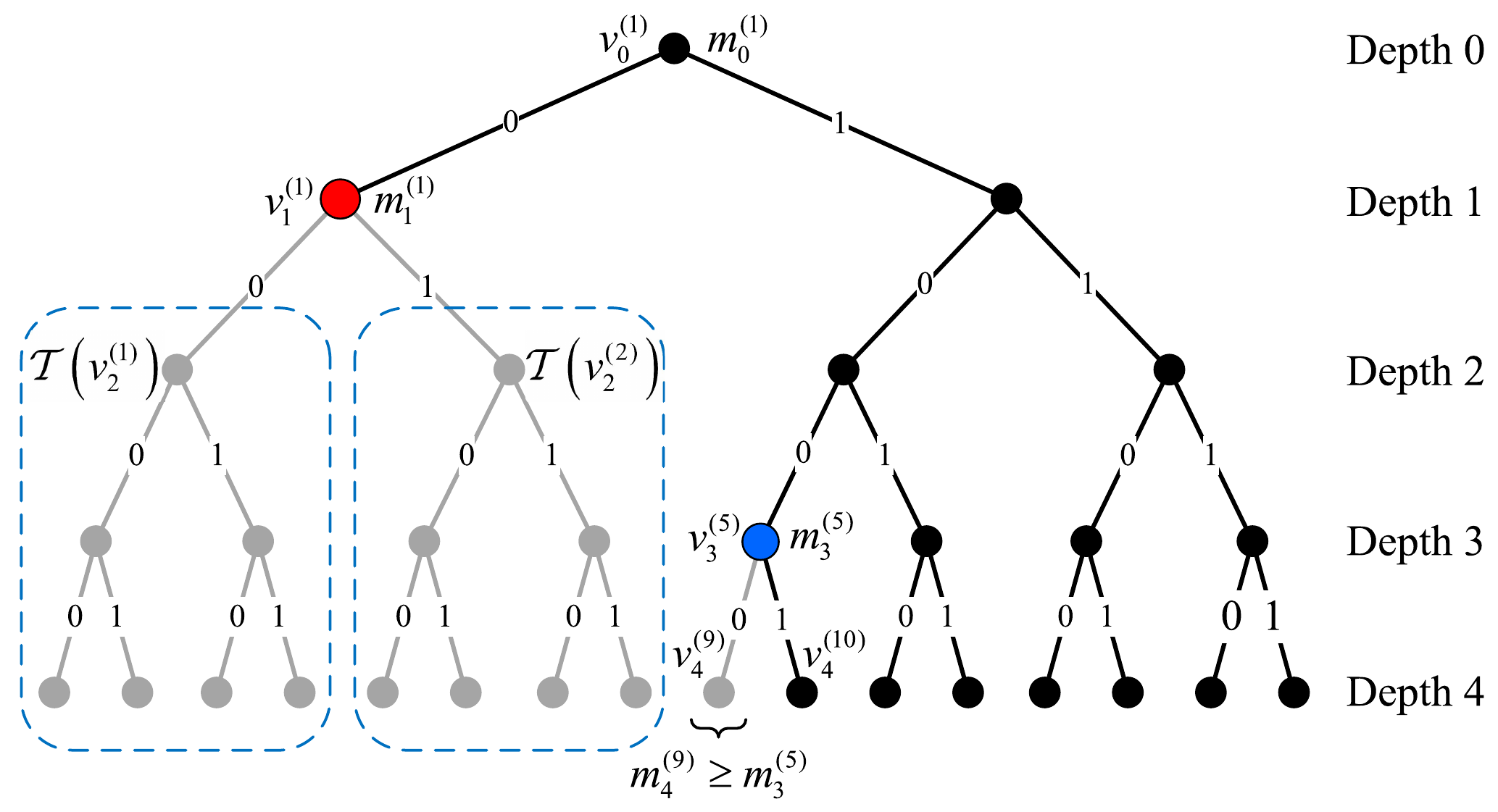}}
  \caption{A code tree representation of polarization violation and polarization reversal, where $N = 16$ and $n = 4$. $m_1^{\left( 1 \right)} \in {{\cal S}_{{\rm{PRS}}}}$ and $m_3^{\left( 5 \right)} \in {{\cal S}_{{\rm{PVS}}}}$.}\label{polar_codetree_PRS}
\vspace{-1em}
\end{figure}

As an example, Fig. \ref{polar_codetree_PRS} shows the polarization violation and polarization reversal in the code tree, where $N = 16$ and $n = 4$. In the AGA computation process, $m_1^{\left( 1 \right)} \in {{\cal S}_{{\rm{PRS}}}}$. Then we have $m_2^{\left( 1 \right)} \ge m_2^{\left( 2 \right)}$, for the leaf nodes with ${{\cal T}_{{\rm{left}}}} = {\cal T}(v_{2}^{\left( {1} \right)})$ and ${{\cal T}_{{\rm{right}}}} = {\cal T}(v_{2}^{\left( {2} \right)})$, following \emph{Theorem 2}, one can check
\begin{equation}\label{example2}
  \begin{array}{l}
    m_4^{\left( 1 \right)} \ge m_4^{\left( 5 \right)},m_4^{\left( 2 \right)} \ge m_4^{\left( 6 \right)},\\
    m_4^{\left( 3 \right)} \ge m_4^{\left( 7 \right)},m_4^{\left( 4 \right)} \ge m_4^{\left( 8 \right)}.
    \end{array}
\end{equation}
Similarly, since $m_3^{\left( 5 \right)} \in {{\cal S}_{{\rm{PVS}}}}$, we have $m_4^{\left( 9 \right)} \ge m_3^{\left( {5} \right)}$, whereas in EGA, it should be $m_4^{\left( 9 \right)} < m_3^{\left( {5} \right)}$, and all the ``$\ge$'' in (\ref{example2}) should be ``$<$''. Thus, compared to EGA, the approximation error of AGA results in obvious rank error within the $16$ leaf nodes.

The above analysis indicates that in the AGA construction of polar codes, if ${{\cal S}_{{\rm{PRS}}}} \ne \varnothing$, its approximation error leads to polarization reversal. Considering its definition, ${{\cal S}_{{\rm{PRS}}}}$ lies in the vicinity of $0$. As stated in \cite{arikan}, when $N$ tends to infinity, the symmetric capacity terms $\{ I(W_N^{(j)}) \} $ cluster around $0$ and $1$, and the corresponding LLR means cluster around $0$ and $ + \infty $. So when the code length becomes longer, there are more LLR means falling in ${{\cal S}_{{\rm{PRS}}}}$ during the AGA recursive computation process. This is the essential reason that polar codes constructed by some AGA suffer from catastrophic performance loss with the long code lengths.

During the recursive process of AGA, assuming the number of code tree nodes belonging to the two sets\footnote{The leaf nodes whose LLR mean falls into the two sets are not counted in, because it will not lead to rank error among the descendants.} are denoted by ${\mu _{{\rm{PVS}}}}$ and ${\mu _{{\rm{PRS}}}}$, the corresponding ratios with respect to all nodes are written as
\begin{equation}\label{Percentage_PVS_PRS}
  {\theta _{{\rm{PVS}}}} = \frac{{{\mu _{{\rm{PVS}}}}}}{{\sum\limits_{k = 0}^{n - 1} {{2^k}} }},~{\theta _{{\rm{PRS}}}} = \frac{{{\mu _{{\rm{PRS}}}}}}{{\sum\limits_{k = 0}^{n - 1} {{2^k}} }}.
\end{equation}
For Chung's conventional AGA, Table \ref{table_AGA} gives the distribution of ${\mu _{{\rm{PVS}}}}$ and ${\mu _{{\rm{PRS}}}}$ with different polarization levels $n$. The corresponding ${\theta _{{\rm{PVS}}}}$ and ${\theta _{{\rm{PRS}}}}$ are also listed in Table \ref{table_AGA}, where ${E_b}/{N_0} = 1{\rm{dB}}$ (${\sigma ^2} = 1.1915$). From Table \ref{table_AGA}, we observe that with the increase of code length, there are more and more nodes falling in PRS, and its corresponding ratio also becomes larger. Therefore, the polar codes constructed by Chung's conventional AGA suffer from catastrophic jitter in performance when their code lengths are long.

\begin{table}[htbp]
\renewcommand\arraystretch{1.3}
\setlength{\abovecaptionskip}{0.cm}
\setlength{\belowcaptionskip}{-0.cm}
  \centering
  \caption{The number of nodes and its percentage whose LLR mean belongs to ${{\cal S}_{{\rm{PVS}}}}$ and ${{\cal S}_{{\rm{PRS}}}}$ during the recursive process of Chung's conventional AGA, where ${E_b}/{N_0} = 1{\rm{dB}}$.}\label{table_AGA}
  \begin{tabular}{|p{1cm}|p{1.3cm}|p{1.3cm}|p{1.3cm}|p{1.3cm}|}
    \hline
    \centering $n$ & \centering ${\mu _{{\rm{PVS}}}}$ & \centering ${\theta _{{\rm{PVS}}}}$ & \centering ${\mu _{{\rm{PRS}}}}$ & \centering ${\theta _{{\rm{PRS}}}}$\tabularnewline
    \hline
    \hline
    \centering $10$ & \centering $40$ & \centering $3.910\%$ & \centering $33$ & \centering $3.226\%$\tabularnewline
    \hline
    \centering $11$ & \centering $89$ & \centering $4.348\%$ & \centering $88$ & \centering $4.299\%$\tabularnewline
    \hline
    \centering $12$ & \centering $191$ & \centering $4.664\%$ & \centering $225$ & \centering $5.495\%$\tabularnewline
    \hline
    \centering $13$ & \centering $394$ & \centering $4.810\%$ & \centering $549$ & \centering $6.702\%$\tabularnewline
    \hline
    \centering $14$ & \centering $803$ & \centering $4.901\%$ & \centering $1297$ & \centering $7.917\%$\tabularnewline
    \hline
    \centering $15$ & \centering $1617$ & \centering $4.935\%$ & \centering $3003$ & \centering $9.165\%$\tabularnewline
    \hline
    \centering $16$ & \centering $3280$ & \centering $5.005\%$ & \centering $6820$ & \centering $10.41\%$\tabularnewline
    \hline
    \centering $17$ & \centering $6340$ & \centering $4.837\%$ & \centering $15240$ & \centering $11.63\%$\tabularnewline
    \hline
    \centering $18$ & \centering $12528$ & \centering $4.779\%$ & \centering $33646$ & \centering $12.83\%$\tabularnewline
    \hline
    \centering $19$ & \centering $24550$ & \centering $4.683\%$ & \centering $73503$ & \centering $14.02\%$\tabularnewline
    \hline
    \centering $20$ & \centering $48036$ & \centering $4.581\%$ & \centering $159132$ & \centering $15.18\%$\tabularnewline
    \hline
  \end{tabular}
\end{table}

Fig. \ref{BLER_1} demonstrates the BLER performance of polar codes constructed by Chung's conventional AGA and Tal\&Vardy's method under the BI-AWGN channels. In Fig. \ref{BLER_1}, polar codes are constructed depending on the signal-to-noise ratio (SNR) one by one, and all the schemes have code rate $R = 1/3$ with SC decoding. The code length $N$ is set to be $2^{12}$, $2^{14}$ and $2^{18}$. We observe that for long code lengths, Chung's scheme obviously presents catastrophic performance loss. It is consistent with the analysis of Table \ref{table_AGA}.

\begin{figure}[htbp]
\vspace{0.4em}
\setlength{\abovecaptionskip}{0.cm}
\setlength{\belowcaptionskip}{-0.cm}
  \centering{\includegraphics[scale=0.61]{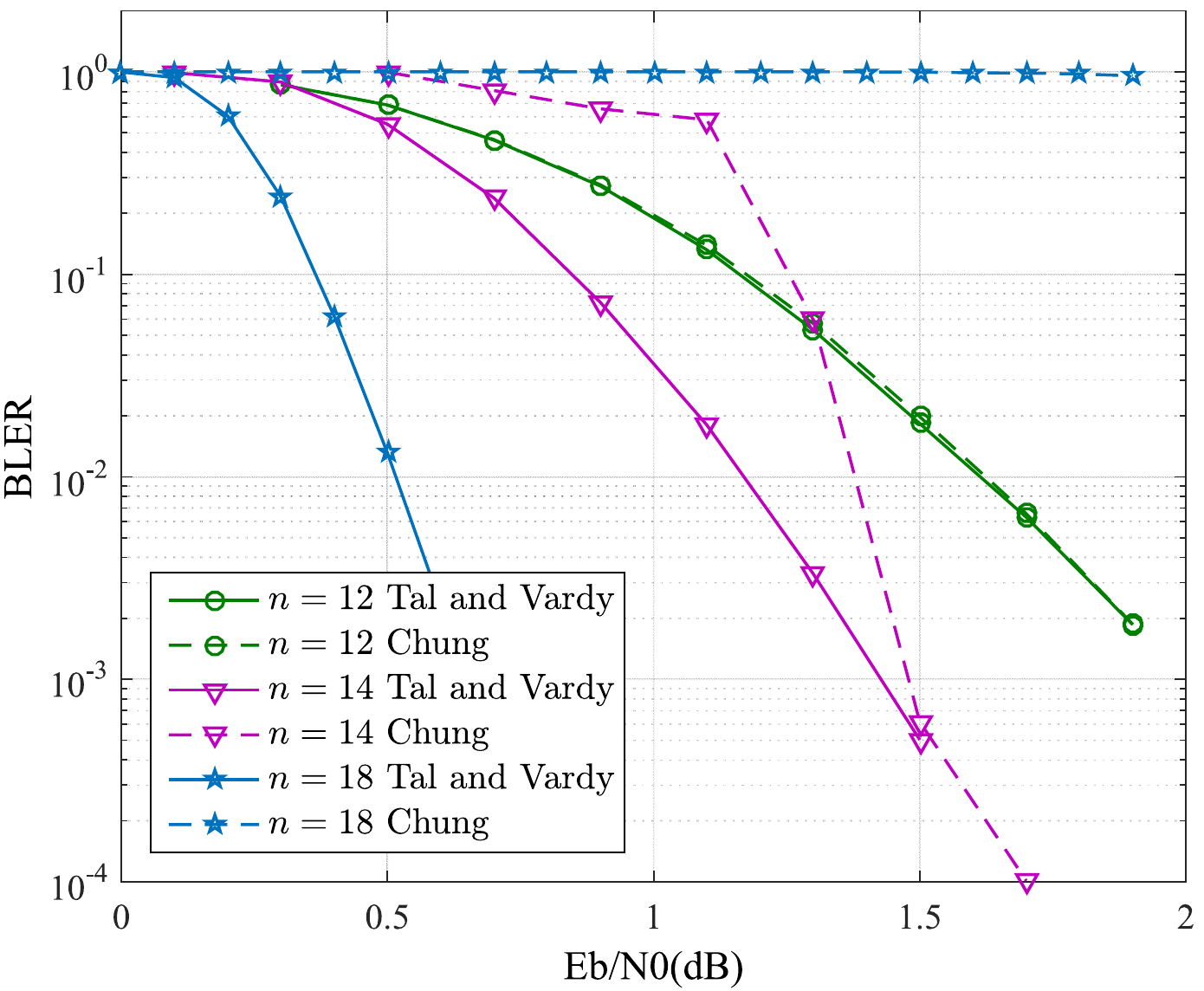}}
  \caption{BLER performance comparison of polar codes with the code length $N = 2^n$ ($n = 12,14,18$) and code rate $R = 1/3$ in the AWGN channel.}\label{BLER_1}
  \vspace{-1em}
\end{figure}

\begin{remark}\label{Remark_1}
\emph{For the AWGN channels, compared with the accurate DE or Tal\&Vardy's algorithm, EGA is also found to well approximate the actual polarized subchannels. Note that EGA has strict order preserving property, which follows \emph{Proposition 2}. This order preserving property of EGA indeed gives reasonable interpretation about its ``good results'' versus DE. Therefore, in general, the error between EGA and DE is so small that it can be ignored.
}
\end{remark}

\begin{remark}\label{Remark_2}
\emph{Recall that Ar{\i}kan in \cite{arikan_rm} suggested a heuristic BEC approximation method to construct the polar codes for arbitrary binary-input channels, which has also yielded good results in experiments. The above PVS and PRS analyses also give interpretation about this ``good results''. One can check BEC approximation has strict order preserving property in size-two polarization transform, which shows $I( {{W_2^{(1)} }} ) < I\left( W \right) < I( {{W_2^{(2)} }} )$. Heuristic BEC approximation will not lead to polarization violation and polarization reversal in this sense. Therefore, this heuristic method will just bring some moderate performance loss rather than catastrophic jitter.
}
\end{remark}


\subsection{CLE of Channel Polarization}
Polarization violation and polarization reversal reveal the essential reason that AGA cannot work well for long code lengths. Besides these two sets, in this subsection, we further propose a new metric, named cumulative-logarithmic error (CLE) of channel polarization, to quantitatively evaluate the remainder approximation error between AGA and EGA. Then CLE is utilized to guide the design of GA approximation function $\Omega \left( t \right)$.

For $\Omega \left( t \right)$ in AGA, suppose its ${{\cal S}_{{\rm{PVS}}}} = \varnothing$ and ${{\cal S}_{{\rm{PRS}}}} = \varnothing$, CLE will play a crucial role in evaluating the performance of AGA. We concern about the subchannel's capacity, which is a function of LLR mean under GA assumption. Hence, the difference between $\Omega \left( t \right)$ and $\phi \left( t \right)$ brings in calculation error in subchannel capacity evaluation. The original absolute error of capacity calculation between AGA and EGA is denoted by ${\Delta \left( t \right)}$, which is a function of LLR mean $t$. Without ambiguity, ${\Delta \left( t \right)}$ will be abbreviated to ${\Delta}$ in this paper.

Assume ${\Delta}$ occurs after $r$ recursions, denoted by ${\Delta _{r}}$. Thus ${\Delta _{r}}$ is accumulated as final error after $n-r$ polarization levels, and this process can be represented on a subtree with a depth $n-r$. To evaluate the calculation error of AGA, we focus on the difference of subchannel's capacity calculated by AGA and EGA in logarithmic domain. The capacities calculated by EGA can form a set $\cal I$ defined on this code subtree with the following properties:

On this subtree, the set of capacities corresponding to the nodes at a given depth $d$ is denoted by ${\cal I}_d$, $d = r,r+1, \cdots ,n$. Let $I_d^{\left( k \right)}$, $k = 1,2, \cdots ,{2^{d-r}}$, denote the $k$-th element in ${\cal I}_d$. For each $I_d^{\left( k \right)} \in {\cal I}_d$, $I_d^{\left( k \right)}$ takes value on $\left[ {0,1} \right]$. For $d > r$, $I_d^{\left( k \right)}$ is a function of path $b_{r+1}^{d}$, which is actually the binary expansion of $k-1$ (i.e., $k - 1 = \sum\nolimits_{i = 1}^{d - r} {{b_{r + i}}{2^{d - r - i}}}$). Therefore, at the root node $v_r^{\left( j \right)}$, ${\Delta _{r}}$ can be written as
\begin{equation}\label{delta_cal}
  {\Delta _r} = \tilde I_r^{\left( 1 \right)} - I_r^{\left( 1 \right)} = h\left( {\frac{2}{{\tilde m_r^{\left( j \right)}}}} \right) - h\left( {\frac{2}{{m_r^{\left( j \right)}}}} \right),
\end{equation}
where $\tilde I_r^{\left( 1 \right)}, I_r^{\left( 1 \right)}$ and ${\tilde m_r^{\left( j \right)}}, { m_r^{\left( j \right)}}$ stand for the initial symmetric capacities and LLR means calculated by AGA and EGA respectively, and the formula of $h\left(  \cdot  \right)$ is written in (\ref{BI_AWGN}).

As stated in \emph{Remark 2}, without much sacrifice in accuracy, BEC approximation will act as faithful surrogate for GA in error analysis. According to the iteration structure in channel polarization transform in \cite{arikan}, we can derive as follows
\begin{equation}\label{cap_iter}
    \left\{{\begin{array}{ll}
    I_{d + 1}^{\left( {2k - 1} \right)} = {\left( {I_{d}^{\left( k \right)}} \right)^2} & {\rm{when}}~{b_{d + 1}} = 0,\\
    I_{d + 1}^{\left( {2k} \right)} = 2I_{d}^{\left( k \right)} - {\left( {I_{d}^{\left( k \right)}} \right)^2} & {\rm{when}}~{b_{d + 1}} = 1.
    \end{array}}\right.
\end{equation}
Furthermore, when $b_{d+1}=1$, we have $I_{d + 1}^{\left( {2k} \right)} \le 2I_{d}^{\left( k \right)}$. Thus, in logarithmic domain, we can get
\begin{equation}\label{log_cap_iter}
    \left\{{\begin{array}{ll}
    {\log}I_{d + 1}^{\left( {2k - 1} \right)} = 2{\log}I_d^{\left( k \right)} & {\rm{when}}~{b_{d + 1}} = 0,\\
    {\log}I_{d + 1}^{\left( {2k} \right)} \le {\log}I_d^{\left( k \right)} + 1 & {\rm{when}}~{b_{d + 1}} = 1.
    \end{array}}\right.
\end{equation}

Define $\tilde I_d^{\left( k \right)} = I_d^{\left( k \right)} + \Delta _d^{\left( k \right)}$, where $\tilde I_d^{\left( k \right)}$ denotes the capacity corresponding to AGA, $\Delta _d^{\left( k \right)}$ represents the absolute error of capacity calculation between AGA and EGA. For $r < d \le n$, $\tilde I_d^{\left( k \right)}$ and $I_d^{\left( k \right)}$ represent the capacities, which are calculated by BEC approximation, of AGA and EGA respectively. In this paper, We only analyze the error between AGA and EGA, rather than the error of GA itself or the error brought in by heuristic BEC approximation. Therefore, $\Delta _{r}^{\left( 1 \right)} = {\Delta _{r}}$. Let $\rho _d^{\left( k \right)} = \Delta _d^{\left( k \right)}/I_d^{\left( k \right)}$ denote the relative error, and $e_d^{\left( k \right)} = {\log}\tilde I_d^{\left( k \right)} - {\log}I_d^{\left( k \right)} = {\log}\left( {1 + \rho _d^{\left( k \right)}} \right)$ denote the capacity calculation error in logarithmic domain. Hence, the partial cumulative-logarithmic error (PCLE) can be written as
\begin{equation}\label{CLEdefine}
    {{C_{r:n}} = \sum\limits_{k = 1}^{2^{n-r}} \left|{e_{n}^{\left( k \right)}}\right|}.
\end{equation}
The cumulative-logarithmic error (CLE) will be $C = \sum\limits_r {{C_{r:n}}} $.

\vspace{-0.65em}
\subsection{CLE Bound}
The precise calculation of CLE is too complicated to be analyzed using recursive relation (\ref{cap_iter}). In this subsection, we propose an upper bound on CLE to simplify its calculation.

\begin{definition}\label{definition_CLE}
\emph{For the $k$-th leaf node corresponding to a path $b_{r+1}^n$, we define $\left| {\left\{ {{r+1} \le i \le n:{b_i} = 0} \right\}} \right| = \alpha$ and $\left| {\left\{ {{r+1} \le i \le n:{b_i} = 1} \right\}} \right| = n - r - \alpha$.
}
\end{definition}

\begin{theorem}\label{theorem_3}
\emph{The logarithm error $\left|e_{n}^{\left( k \right)}\right|$ can be bounded by
\begin{equation*}
    {\left|e_{n}^{\left( k \right)}\right| \le {2^{\alpha}}\left| {{{\log }}\left( {1 + \rho _r^{\left( 1 \right)}} \right)} \right| = {2^\alpha }\left| {{{\log }}\left( {1 + \frac{{\Delta _r^{\left( 1 \right)}}}{{I_r^{\left( 1 \right)}}}} \right)} \right|}.
\end{equation*}
}
\end{theorem}

\begin{proof}
    Let $\tilde e_d^{\left( k \right)}$ denote the bound of $e_d^{\left( k \right)}$. Then, $\tilde e_{n}^{\left( k \right)}$ is determined by specifying $\tilde e_r^{\left( 1 \right)} = e_r^{\left( 1 \right)} = {\log}\left( {1 + \rho _r^{\left( 1 \right)}} \right)$ and
    \begin{equation}\label{recursion}
        \left\{{\begin{array}{ll}
        \tilde e_{d + 1}^{\left( {2k - 1} \right)} = D\left( {\tilde e_d^{\left( k \right)}} \right) & {\rm{when}}~{b_{d + 1}} = 0,\\
        \tilde e_{d + 1}^{\left( {2k} \right)} = E\left( {\tilde e_d^{\left( k \right)}} \right) & {\rm{when}}~{b_{d + 1}} = 1,
        \end{array}}\right.
    \end{equation}
    where $E$ : ${\mathbb{R}} \to {\mathbb{R}}$, $E\left( z \right) = z$ denotes equality, $D$ : ${\mathbb{R}} \to {\mathbb{R}}$, $D\left( z \right) = 2z$ denotes doubling.

     If $\Delta _r^{\left( 1 \right)} \ge 0$, it claims that $\Delta _d^{\left( k \right)} \ge 0$ holds in terms of \emph{Theorem \ref{Theorem_2}}. Note that during the iteration, when ${b_{d + 1}} = 0$,
      \begin{equation}\label{ineq_1}
        {0 \le e_{d + 1}^{\left( {2k - 1} \right)} = 2e_d^{\left( k \right)} = \tilde e_{d + 1}^{\left( {2k - 1} \right)}}.
      \end{equation}
      When ${b_{d + 1}} = 1$, it can be proved that $e_{d + 1}^{\left( {2k} \right)} \le \tilde e_{d + 1}^{\left( {2k} \right)}$. According to the first equation of (\ref{cap_iter}), it is easy to get that
    \begin{equation}
        {\begin{aligned}
        \tilde I_{d + 1}^{\left( {2k} \right)} & = 2\tilde I_d^{\left( k \right)} - {\left( {\tilde I_d^{\left( k \right)}} \right)^2}\\
        ~ & = 2\left( {I_d^{\left( k \right)} + \Delta _d^{\left( k \right)}} \right) - {\left( {I_d^{\left( k \right)} + \Delta _d^{\left( k \right)}} \right)^2}\\
        ~ & = \underbrace {2I_d^{\left( k \right)} - {{\left( {I_d^{\left( k \right)}} \right)}^2}}_{= I_{d + 1}^{\left( {2k} \right)}} + \underbrace {2\Delta _d^{\left( k \right)} - 2I_d^{\left( k\right)}\Delta _d^{\left( k \right)} - {{\left( {\Delta _d^{\left( k \right)}} \right)}^2}}_{= \Delta _{d + 1}^{\left( {2k} \right)}}.
        \end{aligned}}
    \end{equation}
    Therefore, $e_{d + 1}^{\left( {2k} \right)}$ can be written as
    \begin{equation}
       \begin{array}{ll}
        e_{d + 1}^{\left( {2k} \right)} & = {\log}\left( {1 + \rho _{d + 1}^{\left( {2k} \right)}} \right)\\
         ~ & = {\log}\left( {1 + \frac{{2\Delta _d^{\left( k \right)} - 2I_d^{\left( k \right)}\Delta _d^{\left( k \right)} - {{\left( {\Delta _d^{\left( k \right)}} \right)}^2}}}{{2I_d^{\left( k \right)} - {{\left( {I_d^{\left( k \right)}} \right)}^2}}}} \right).
        \end{array}
    \end{equation}
    In addition, we have
    \begin{equation}
        \tilde e_{d + 1}^{\left( {2k} \right)} = {\log}\left( {1 + \tilde \rho _{d + 1}^{\left( {2k} \right)}} \right) = e_d^{\left( k \right)} = {\log}\left( {1 + \frac{{\Delta _d^{\left( k \right)}}}{{I_d^{\left( k \right)}}}} \right).
    \end{equation}
     Then we can check
    \begin{equation}\label{forCLEupbound}
        {\tilde \rho _{d + 1}^{\left( {2k} \right)} - \rho _{d + 1}^{\left( {2k} \right)} = \frac{{\Delta _d^{\left( k \right)}\left( {I_d^{\left( k \right)} + \Delta _d^{\left( k \right)}} \right)}}{{I_d^{\left( k \right)}\left( {2 - I_d^{\left( k \right)}} \right)}} \ge 0}.
    \end{equation}
    It can be inferred that $\tilde e_d^{\left( k \right)} \ge e_d^{\left( k \right)} \ge 0$ holds.

    Recall \emph{Definition \ref{definition_CLE}}, during the iterative calculation of $\tilde e_{n}^{\left( k \right)}$, we count doubling $\alpha$ times and equality $n-r-\alpha$ times. Hence, we have
    \begin{equation}\label{boundpositive}
        {0 \le e_{n}^{\left( k \right)} \le \tilde e_{n}^{\left( k \right)} = {E^{n - r - \alpha }}{D^\alpha }\left( {e_r^{\left( 1 \right)}} \right) = {2^\alpha}e_r^{\left( 1 \right)} }.
    \end{equation}

    Analgously, if $\Delta _r^{\left( 1 \right)} < 0$, we have $\Delta _d^{\left( k \right)} < 0$. From (\ref{ineq_1}) and (\ref{forCLEupbound}), we can get
    \begin{equation}\label{boundnegative}
        {0 > e_{n}^{\left( k \right)} > \tilde e_{n}^{\left( k \right)} = {E^{n - r - \alpha }}{D^\alpha }\left( {e_r^{\left( 1 \right)}} \right) = {2^\alpha}e_r^{\left( 1 \right)} }.
    \end{equation}

    Combing (\ref{boundpositive}) and (\ref{boundnegative}), we prove the theorem.
\end{proof}

\begin{theorem}\label{Theorem_4}
\emph{PCLE ${C_{r:n}}$ can be upper bounded by
\begin{equation*}
    {{C_{r:n}} \le {3^{n-r}}\left|{\log}\left( {1 + \rho _r^{\left( 1 \right)}} \right)\right| = {3^{n-r}}\left| {{{\log }} \left( {1 + \frac{\Delta _r^{\left( 1 \right)}}{I_r^{\left( 1 \right)}}} \right)} \right|}.
\end{equation*}
}
\end{theorem}

\begin{proof}
    For any $k \in \left\{ {1,2, \cdots ,{2^{n-r}}} \right\}$, the number of vectors $b_{r+1}^n$ satisfying \emph{Definition \ref{definition_CLE}} is $\tbinom{n-r}{\alpha}$, where $b_{r+1}^n$ is the binary expansion of $k-1$. Combined with definition (\ref{CLEdefine}) and \emph{Theorem \ref{theorem_3}}, $C_{r:n}$ satisfies the following constraint
    \begin{equation}\label{CLEupbound}
        {C_{r:n}} \le \sum\limits_{k = 1}^{{2^{n - r}}} {\left| {\tilde e_{n}^{\left( k \right)}} \right|}  = \sum\limits_{\alpha  = 0}^{n - r} {\tbinom{n-r}{\alpha} {2^\alpha }\left| {e_r^{\left( 1 \right)}} \right|}  = {3^{n - r}}\left| {e_r^{\left( 1 \right)}} \right|.
    \end{equation}
    The last equation in (\ref{CLEupbound}) uses binomial theorem.
    Therefore, CLE $C$ can be upper bounded by $\sum\limits_r {{3^{n - r}}\left| {{{\log }_2}\left( {1 + \rho _r^{\left( 1 \right)}} \right)} \right|}$, and the exponent $n - r$ stands for polarization levels.
\end{proof}

\section{Improved Gaussian Approximation}
In this section, guided by the previous PVS, PRS and CLE analyses, we propose new rules to design AGA for the construction of polar codes. Then we give three specific forms of the approximation function in AGA, which have advantages in both complexity and performance.

\vspace{-0.5em}
\subsection{Design Rules of AGA}
For channels other than BEC (e.g. AWGN channel), AGA is widely used to construct polar codes. However, in practical implementation, the accuracy of approximation function $\Omega \left( t \right)$ will greatly affect the construction of polar codes especially when the code length is long. According to \emph{Theorem \ref{Theorem_4}}, the initial error will be exponentially amplified with the increase of polarization levels. Note that PCLE bound is mainly affected by two factors: the first term $3^{n-r}$ is relevant to polarization levels, and the second term is dependent on the original relative error ${\rho _r^{\left( 1 \right)}} = \Delta _r^{\left( 1 \right)}/I_r^{\left( 1 \right)}$. In general, the absolute error $\Delta _r^{\left( 1 \right)}$ is tiny. Therefore, for the good channels whose capacities $I_r^{\left( 1 \right)}$ approach $1$, their original relative errors ${\rho _r^{\left( 1 \right)}}$ are so small that they can be ignored. However, for the bad channels whose capacities $I_r^{\left( 1 \right)}$ approach $0$, their original relative errors are not negligible. Subsequently, given a fixed absolute error $\Delta _r^{\left( 1 \right)}$, the more $I_r^{\left( 1 \right)}$ is close to $0$, the larger the original relative error ${\rho _r^{\left( 1 \right)}}$ becomes.

Above analysis indicates that CLE bound is mainly affected by the terms $C_{r:n}$ with the small initial capacity $I_r^{\left( 1 \right)}$, which corresponds to the bad channels. Due to the error, some frozen subchannels will be wrongly identified as information-carrying ones (role flipping), which results in the performance degradation. Note that the capacity $I_r^{\left( 1 \right)}$ monotonically increases with LLR mean $t$ under the GA assumption. In addition, according to (\ref{capacity_limits}), we have
\begin{equation}
  \mathop {\lim }\limits_{t \to 0} I_r^{\left( 1 \right)} = 0, \mathop {\lim }\limits_{t \to +\infty } I_r^{\left( 1 \right)} = 1.
\end{equation}

Guided by above analyses, the AGA approximation function design is composed of three rules:

\begin{enumerate}[Rule 1)]
  \item PVS and PRS eliminating: $\Omega \left( t \right)$ should guarantee ${{\cal S}_{{\rm{PVS}}}} = \varnothing$ and ${{\cal S}_{{\rm{PRS}}}} = \varnothing$. According to \emph{Proposition \ref{proposition_PVS_PRS_relationship}} and its converse-negative proposition, if ${{\cal S}_{{\rm{PVS}}}} = \varnothing$, we have ${{\cal S}_{{\rm{PRS}}}} = \varnothing$. Hence, in order to empty PVS and PRS, we should ensure $0 < \Omega \left( t \right) < 1$ for any $t \in \left( {0, + \infty } \right)$.
  \item Low SNR design: When $t$ comes close to $0$, we must guarantee $\mathop {\lim }\limits_{t \to 0} \Omega \left( t \right) = 1$ to reduce approximation error. Since CLE bound is amplified exponentially with the growth of polarization levels, the only way to reduce CLE bound is to lower the initial relative error ${\rho _r^{\left( 1 \right)}}$. Therefore, $\Omega \left( t \right)$ needs to be divided into more segments when $t$ approaches $0$. This rule can reduce the original absolute error $\Delta _r^{\left( 1 \right)}$ in the vicinity of $0$ so as to lower ${\rho _r^{\left( 1 \right)}}$.
  \item  High SNR design: When $t$ stays away from $0$, thanks to the relatively large $I_r^{\left( 1 \right)}$, CLE bound can tolerate a more obvious absolute error $\Delta _r^{\left( 1 \right)}$. Therefore, $\Omega \left( t \right)$ can be selected with some simpler forms to reduce the computational complexity.
\end{enumerate}

Besides, $\Omega \left( t \right)$ should keep continuity between the adjacent two segments, which can mitigate the jitter of CLE bound by keeping smooth of initial error. Among these three rules, Rule 1 is the most crucial, which can prevent the corresponding AGA polar code construction scheme from severe performance loss. Then the importance of Rule 2 is second, which reduces the remainder approximation error between AGA and EGA. Rule 3 plays a less important role, which helps to further reduce the computational complexity of AGA.

\vspace{-0.5em}
\subsection{Improved AGA Approximation Function}
Recall the analysis in Section IV, Chung's two-segment AGA scheme leads to ${{\cal S}_{{\rm{PVS}}}} \ne \varnothing$ and ${{\cal S}_{{\rm{PRS}}}} \ne \varnothing$, which violates Rule 1. Therefore, guided by above rules, we design the following new two-segment approximation function by employing curve-fitting algorithm with minimum mean square error criterion, denoted by ${\Omega _2}\left( t \right)$,
\begin{equation}\label{AGA_2}
    {\Omega _2}\left( t \right) = \left\{ {\begin{array}{*{20}{ll}}
    {{e}^{0.0116{t^2} - 0.4212t}} & {0 < t \le a,}\\
    {{e}^{-0.2944t - 0.3169}} & {a < t,}\\
    \end{array}} \right.
 \end{equation}
where the boundary point $a = 7.0633$. The corresponding new AGA algorithm is denoted by AGA-2. In the construction of polar codes, compared with Chung's scheme, AGA-2 can eliminate rank error and therefore will not lead to catastrophic performance loss at the long code lengths. In addition, since the inverse function of ${\Omega _2}\left( t \right)$ can be analytically obtained, it has lower calculating complexity. Furthermore, guided by Rule 3, for ${f_c}\left( t \right)$ in (\ref{f1f2}), when $t$ leaves away from $0$, ${1 - {\Omega _2}\left( t \right)}$ tends to $1$, which indicates
\begin{equation}\label{simp_phi}
  {\left( {1 - {\Omega _2}\left( t \right)} \right)^2} \approx 1 - {\Omega _2}\left( t \right) \Rightarrow {f_c}\left( t \right) \approx t.
\end{equation}
Followed by \emph{Proposition 2}, ${f_c}\left( t \right)$ should satisfy ${f_c}\left( t \right) < t$. Thus, when $t$ stays away from $0$, the complex ${f_c}\left( t \right)$ can be further approximated as
\begin{equation}
  {f_c}\left( t \right) = t - \varepsilon,
\end{equation}
where $\varepsilon$ denotes the offset. Then according to (\ref{f1f2}), when ${\Omega _2}\left( t \right)$ tends to $0$, one can check
\begin{equation}\label{AGA_simple}
  {\Omega _2} \left( {t - \varepsilon } \right) = 2{\Omega _2} \left( t \right) - {\Omega _2} {\left( t \right)^2} \approx 2{\Omega _2} \left( t \right).
\end{equation}
Therefore, for ${\Omega _2}\left( t \right)$, in terms of (\ref{AGA_2}), when $t \gg 0$ we have
\begin{equation}
  \begin{aligned}
    {{e}^{ - 0.2944\left( {t - \varepsilon } \right) - 0.3169}} & = {{e}^{ - 0.2944t - 0.3169 + \ln 2}}\\
     \Rightarrow \varepsilon & = 2.3544,
    \end{aligned}
\end{equation}
where $t - \varepsilon $ should locate in the sencond segment, namely $t - \varepsilon  > a$, which claims $t > 9.4177$. Following Rule 3, for the entire AGA-2 scheme, its ${f_c}\left( t \right)$ is further simplified as
\begin{equation}\label{AGA3_final}
  {f_c}\left( t \right) = \left\{ {\begin{array}{*{20}{ll}}
    {{{\Omega}_2^{ - 1}}\left( {1 - {{\left( {1 - {\Omega _2}\left( t \right)} \right)}^2}} \right)}&{0 \le t \le \tau ,}\\
    {t - 2.3544}&{t > \tau ,}
    \end{array}} \right.
\end{equation}
where the boundary point is $\tau = 9.4177$.

From the subsequent CLE analysis, we find that although AGA-2 can satisfy Rule 1 so as to remove the rank error, it will still bring obvious approximation error in the vicinity of $0$. Therefore, when the code length is relatively long, AGA-2 will also bring some moderate performance loss. In order to further improve the GA construction performance, we propose a new piecewise function ${\Omega _3}\left( t \right)$ with three segments by using curve-fitting algorithm, that is
\begin{equation}\label{AGA_3}
    {\Omega _3}\left( t \right) = \left\{ {\begin{array}{*{20}{ll}}
    {{e}^{0.06725{t^2} - 0.4908t}} & {0 < t \le a,}\\
    {{{\mathop{e}\nolimits} ^{ - 0.4527{t^{0.86}} + 0.0218}}}&{a < t \le b,}\\
    {{e}^{ - 0.2832t - 0.4254}} & {b < t,}
    \end{array}} \right.
\end{equation}
where the boundary points $a = 0.6357$ and $b = 9.2254$. Its corresponding AGA algorithm is denoted by AGA-3. It is specially designed for polar codes, which follows the proposed rules. AGA-3 can better satisfy Rule 1 and Rule 2. Especially in the third segment, namely the high SNR region, AGA-3 has lower computational complexity compared to Chung's scheme which obeys Rule 3. In addition, the inverse function of ${\Omega _3}\left( t \right)$ can be easily obtained. Similar with AGA-2, for the whole AGA-3 scheme, its ${f_c}\left( t \right)$ can be further simplified as
\begin{equation}\label{AGA3_final_function}
  {f_c}\left( t \right) = \left\{ {\begin{array}{*{20}{ll}}
    {{{\Omega}_3^{ - 1}}\left( {1 - {{\left( {1 - {\Omega _3}\left( t \right)} \right)}^2}} \right)}&{0 \le t \le \tau ,}\\
    {t - 2.4476}&{t > \tau ,}
    \end{array}} \right.
\end{equation}
where the boundary point is $\tau = 11.673$.

There is no doubt that, if the code length becomes extremely long, the three-segment approximation function in AGA-3 will also bring calculation error, which obeys  \emph{Theorem \ref{Theorem_4}} and is stated in Rule 2. Hence, we design the following four-segment improved AGA scheme for extremely long polar code, that is
\begin{equation}\label{AGA_4}
    {\Omega _4}\left( t \right) = \left\{ {\begin{array}{*{20}{ll}}
    {{e}^{0.1047{t^2} - 0.4992t}} & {0 < t \le a,}\\
    {0.9981\cdot{e}^{0.05315{t^2} - 0.4795t}} & {a < t \le b,}\\
    {{{\mathop{e}\nolimits} ^{ - 0.4527{t^{0.86}} + 0.0218}}}&{b < t \le c,}\\
    {{e}^{ - 0.2832t - 0.4254}} & {c < t,}
    \end{array}} \right.
\end{equation}
where the boundary points $a = 0.1910$, $b = 0.7420$ and $c = 9.2254$. The approximation accuracy in the vicinity of $0$ is further improved. Its corresponding AGA algorithm is denoted by AGA-4. The inverse function of ${\Omega _4}\left( t \right)$ can also be analytically obtained. Since the last two segments in ${\Omega _4}\left( t \right)$ are the same as that in ${\Omega _3}\left( t \right)$, the ${f_c}\left( t \right)$ function in AGA-4 also has the same form as the equation (\ref{AGA3_final_function}).

The computational complexity of ${f_c}\left( t \right)$ has been remarkably reduced thanks to Rule 3 in the above three proposed AGA schemes. Moreover, Rule 1 and Rule 2 help AGA to achieve excellent performance. Thus, for the general construction of polar codes, the proposed AGA-2 scheme is a good alternate to improve the Chung's conventional AGA scheme. When the code length becomes longer, the proposed AGA-3 scheme will achieve better performance than AGA-2 scheme. If the code length becomes extremely long, AGA-4 will present better performance. When the code length continues to grow, AGA-4 will inevitably bring calculation error, which follows \emph{Theorem \ref{Theorem_4}} and Rule 2. Nevertheless, followed by Rule 1 to 3, we can still design AGA with more complex multi-segment approximation function to keep the calculation accuracy. In this way, we provide a systematic framework to design high accuracy AGA scheme for polar codes at any code length.

As for the specific forms of the new AGA approximation functions, it is heuristically obtained from Chung's original $\varphi\left( t \right)$ while taking the Rule 1 to 3 into consideration. The first segment should ensure $\mathop {\lim }\limits_{t \to 0} {\Omega}\left( t \right) = 1$ and use some complex functions to reduce the relative error. When $t$ stays away from $0$, $\Omega \left( t \right)$ can be selected with some simpler forms to reduce the computational complexity. After choosing the approximation form, its parameters are acquired by curve-fitting. Certainly, one can choose different approximation functions according to the proposed rules. Finally, the performance will be evaluated by calculating its CLE bound.

For other channels $\overline W $ rather than the AWGN (e.g. the Rayleigh fading channel), one may expect better performance at the expense of more complexity in the code construction by using DE. However, we can approximate the channel $\overline W $ using an AWGN channel $W$ with $\sigma $, where
\begin{equation}
  I\left( {\overline W } \right) = I\left( W \right) = h\left( {{\sigma ^2}} \right).
\end{equation}
The code construction is then then performed over each of the equivalent AWGN channels in the same GA way as that in the conventional AWGN case. As that will be shown in Section VI, the proposed AGA based construction of polar codes works well for other channels as well. This is significant in that it shows the robustness of the AGA construction against uncertainty and variation in channel parameters.

\vspace{-0.5em}
\subsection{Complexity Comparison}
In this part, we compare the order of complexities with four typical polar code construction methods, which are Ar{\i}kan's heuristic BEC approximation method, DE algorithm, Tal and Vardy's method and our proposed AGA schemes. Note that the complexity order provided only includes the dominant terms, and the detailed number of calculation depends heavily on the specific hardware implementation, which is beyond the scope of this paper.

One can check the principle of these four methods is to recursively calculate each subchannel's reliability on the code tree. Therefore, their the total number of node visiting complexity can be written as $O\left( N \right) + O\left( {N/2} \right) + O\left( {N/4} \right) +  \cdots  + O\left( 2 \right) + O\left( 1 \right) = O\left( N \right)$. It seems that these four methods have the same number of node visiting. Furthermore, we analyze the computational complexity of each visiting.
\begin{itemize}
  \item For the heuristic BEC approximation, the computational complexity of each visiting operation is $O\left( 1 \right)$, which can be ignored. Hence, it has the lowest complexity among the four schemes.
  \item DE needs a fast Fourier transform (FFT) and an inverse fast Fourier transform (IFFT) operation when calculating the probability density function (PDF) of LLR of each bit. The corresponding complexity is $O\left( {\xi \log \xi } \right)$, where $\xi$ denotes the number of samples. However, a typical value of $\xi$ is about $10^5$ \cite{niukai}, implying a huge computational burden in practical application. The difficulty is further aggravated by the quantization errors, which will be accumulated over multiple polarization levels.
  \item Tal and Vardy's method can be viewed as an approximate version of DE \cite{talvardy} and has a complexity of $O( {{\mu ^2}\log \mu } )$, where $\mu$ is a fixed, even and positive integer independent of code length $N$. In general, a typical value of $\mu$ is $256$ which is much less than $\xi$. Hence, this method has much lower complexity than DE. But when the code length becomes long, Tal and Vardy's method still involves relatively high computational complexity.
  \item Regarding the proposed AGA method, the complexity of calculating ${f_v}\left( t \right)$ can be ignored since the result can be obtained easily. For ${f_c}\left( t \right)$, when $t > \tau$, its calculating complexity can also be ignored. When $0 \le t \le \tau$, since the inverse function in ${f_c}\left( t \right)$ can be analytically obtained, the computational complexity is $O\left( 1 \right)$. From these comparison, it can be concluded that our proposed AGA schemes have similar calculating complexity order with the heuristic BEC approximation, which are much lower than DE or Tal and Vardy's algorithms.
\end{itemize}


\section{Performance Evaluation}

\begin{figure}[b]
\setlength{\abovecaptionskip}{0.cm}
\setlength{\belowcaptionskip}{-0.cm}
  \centering{\includegraphics[scale=0.64]{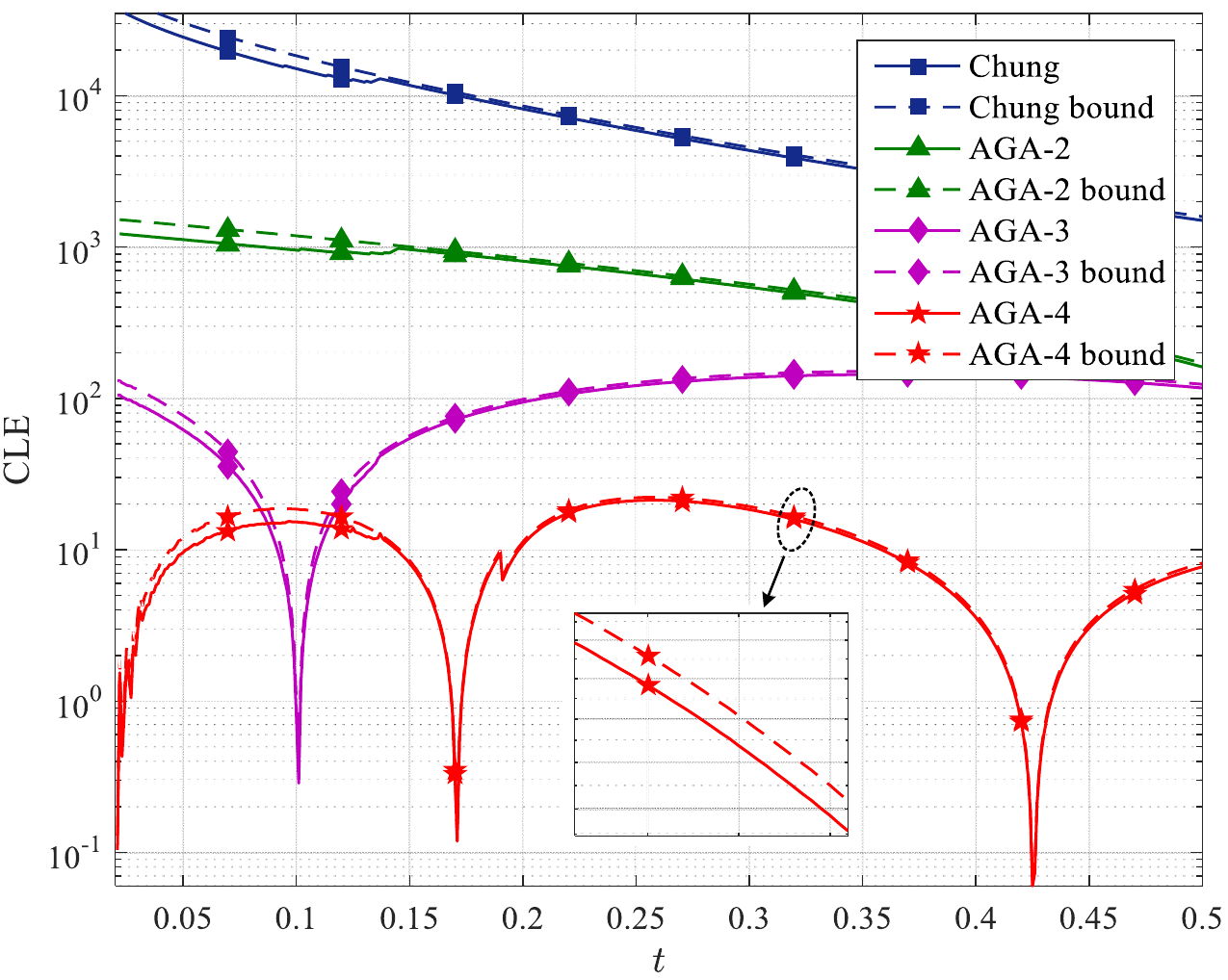}}
  \caption{The precise results and corresponding upper bounds of CLE for various AGAs, where $t$ stands for LLR mean and polarization level $n = 8$.}\label{CLE_plot_1}
\end{figure}

The precise results and corresponding upper bounds of CLE for various AGAs are depicted in Fig. \ref{CLE_plot_1}, where the polarization level $n = 8$ and $r = 0$ (CLE $C = {C_{0:8}}$). It can be found that the CLE bound and the exact result coincide well. Therefore, the CLE bound may be used as an effective tool to evaluate the performance of various AGA schemes. From Fig. \ref{CLE_plot_1}, we can see that the CLE bound of Chung's conventional AGA scheme is obviously higher than that of AGA-2 and AGA-3, which indicates the poor performance of polar codes when Chung's AGA scheme is used. Compared with AGA-2, the CLE bound of AGA-3 also shows some performance gain. Furthermore, AGA-4 can achieve the best performance among these AGA schemes. We notice that there exists a non-monotonic behavior for the CLE bound. This phenomenon is caused by function fitting because there is different relative error $\rho _r^{\left( 1 \right)}$ for different $t$, which results in the jitter of CLE bound.

Next we compare the BLER performance among different construction schemes under the BI-AWGN channel, which is shown in Fig. \ref{BLER_plot_1}. All the schemes have code rate $R = 1/3$ with the SC decoding. The code lengths $N$ are set to $2^{12}$, $2^{14}$ and $2^{18}$, respectively. From these results, it can be found that the BEC approximation scheme demonstrates some moderate performance loss compared with other advanced construction methods. For the extremely long code lengths, since DE falls into a huge computational burden in the practical application, we use Tal and Vardy's method as an alternate with almost no performance loss. Therefore, Tal and Vardy's construction possesses the highest accuracy in Fig. \ref{BLER_plot_1}. However, it still requires high computational complexity at long code lengths. Among the AGA schemes, we can see that Chung's scheme suffers from a dramatic performance loss with the increase of code length since it violates Rule 1. On the contrast, our proposed AGA-2 achieves good performance. Hence, AGA-2 can be employed as a good alternate for Chung's two-segment method at some moderate code lengths. It can also be observed that the proposed AGA-3 scheme achieve better performance which follows Rule 1 and Rule 2. Furthermore, we can see that AGA-4 scheme approaches the performance of Tal and Vardy's method, which can be used for some extremely long code length. In addition, with the increase of code length, the performance gap between AGA-3 and AGA-4 becomes larger. Therefore, in terms of the previous CLE bound analyses and simulation results, it can be predicted that the performance gap will become more and more obvious with the increase of code length.

The comprehensive BLER performance comparisons with the SC decoding in terms of different construction schemes under the Rayleigh fading channels are given in Fig. \ref{BLER_plot_2}. The code length is $N = 2^{14}$, and the code rates $R$ are set to $1/3$ and $2/5$, respectively. Similar with Fig. \ref{BLER_plot_1}, we can see that DE algorithm achieves the best performance. Chung's scheme also presents poor performance since it cannot satisfy Rule 1. On the contrary, AGA-4 scheme suffers from an ignorable loss of performance compared with DE algorithm. Furthermore, it becomes more computationally efficient and implementable in practical use than the former. In addition, AGA-4 can stably achieve $0.15\sim0.2$dB gain for different code rates compared with AGA-3 scheme. These results shows the robustness of the AGA construction against uncertainty and variation in channel parameters, which is valuable for polar coding.

\begin{figure}[htbp]
\setlength{\abovecaptionskip}{0.cm}
\setlength{\belowcaptionskip}{-0.cm}
  \centering{\includegraphics[scale=0.63]{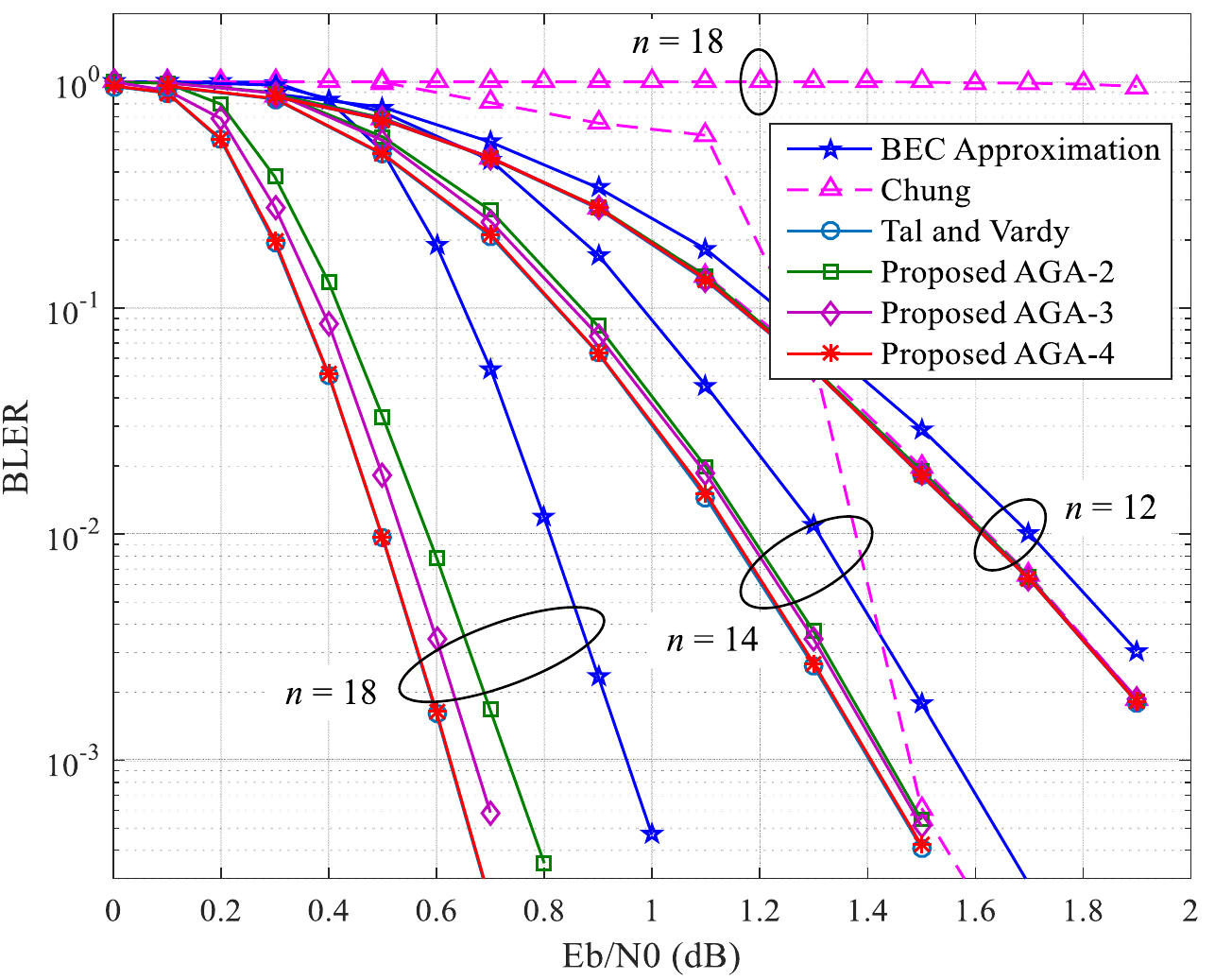}}
  \caption{The SC decoding BLER performance comparison of polar codes with the code length $N = 2^n$ ($n = 12,14,18$) and the code rate $R = 1/3$ over the BI-AWGN channel.}\label{BLER_plot_1}
  \vspace{-1em}
\end{figure}

\begin{figure}[htbp]
\setlength{\abovecaptionskip}{0.cm}
\setlength{\belowcaptionskip}{-0.cm}
  \centering{\includegraphics[scale=0.54]{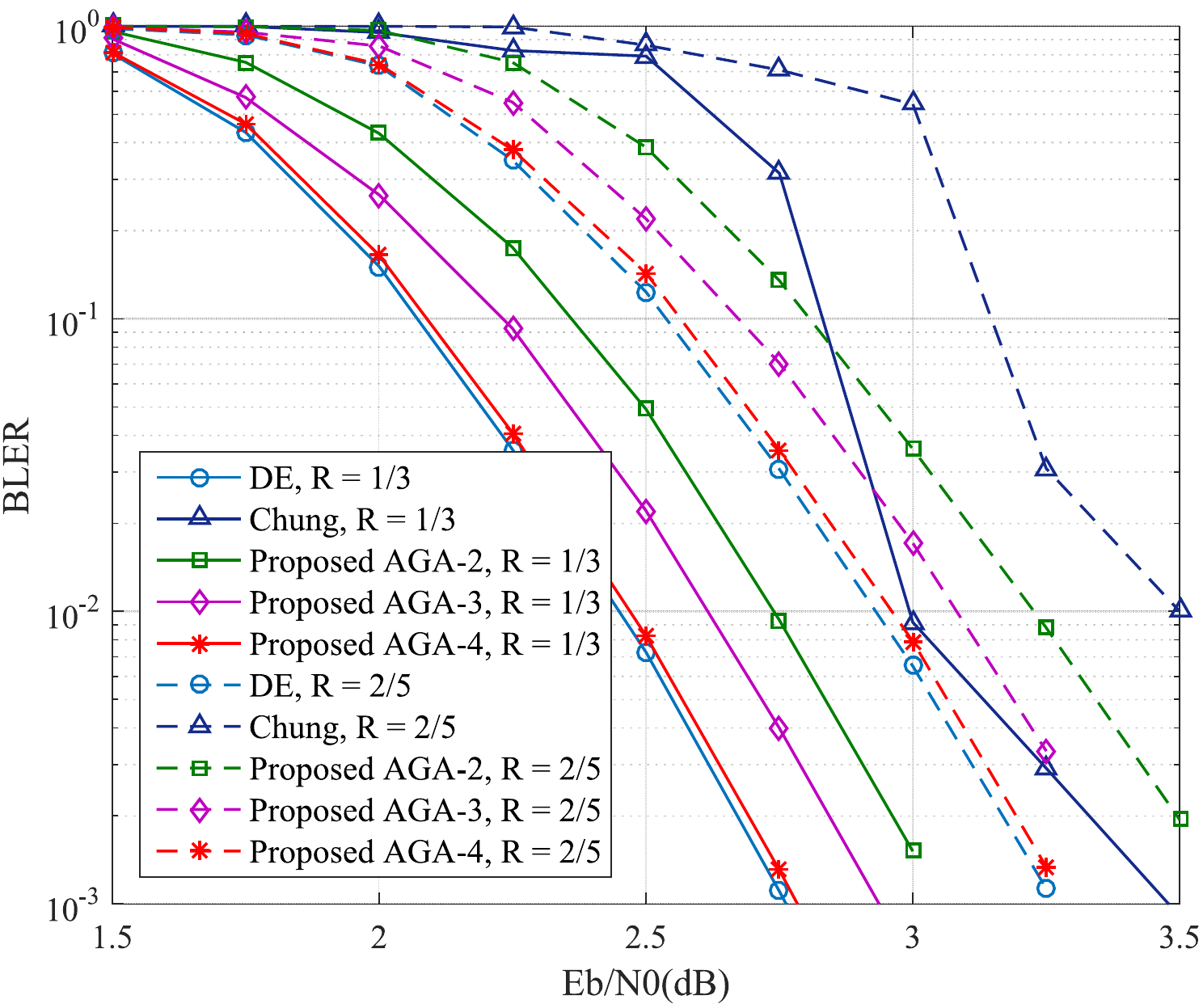}}
  \caption{The SC decoding BLER performance comparison of polar codes with code length $N = 2^{14} = 16384$ and code rates $R = 1/3$ and $R = 2/5$ over the Rayleigh fading channel.}\label{BLER_plot_2}
  \vspace{-1em}
\end{figure}

\begin{figure}[htbp]
\setlength{\abovecaptionskip}{0.cm}
\setlength{\belowcaptionskip}{-0.cm}
  \centering{\includegraphics[scale=0.54]{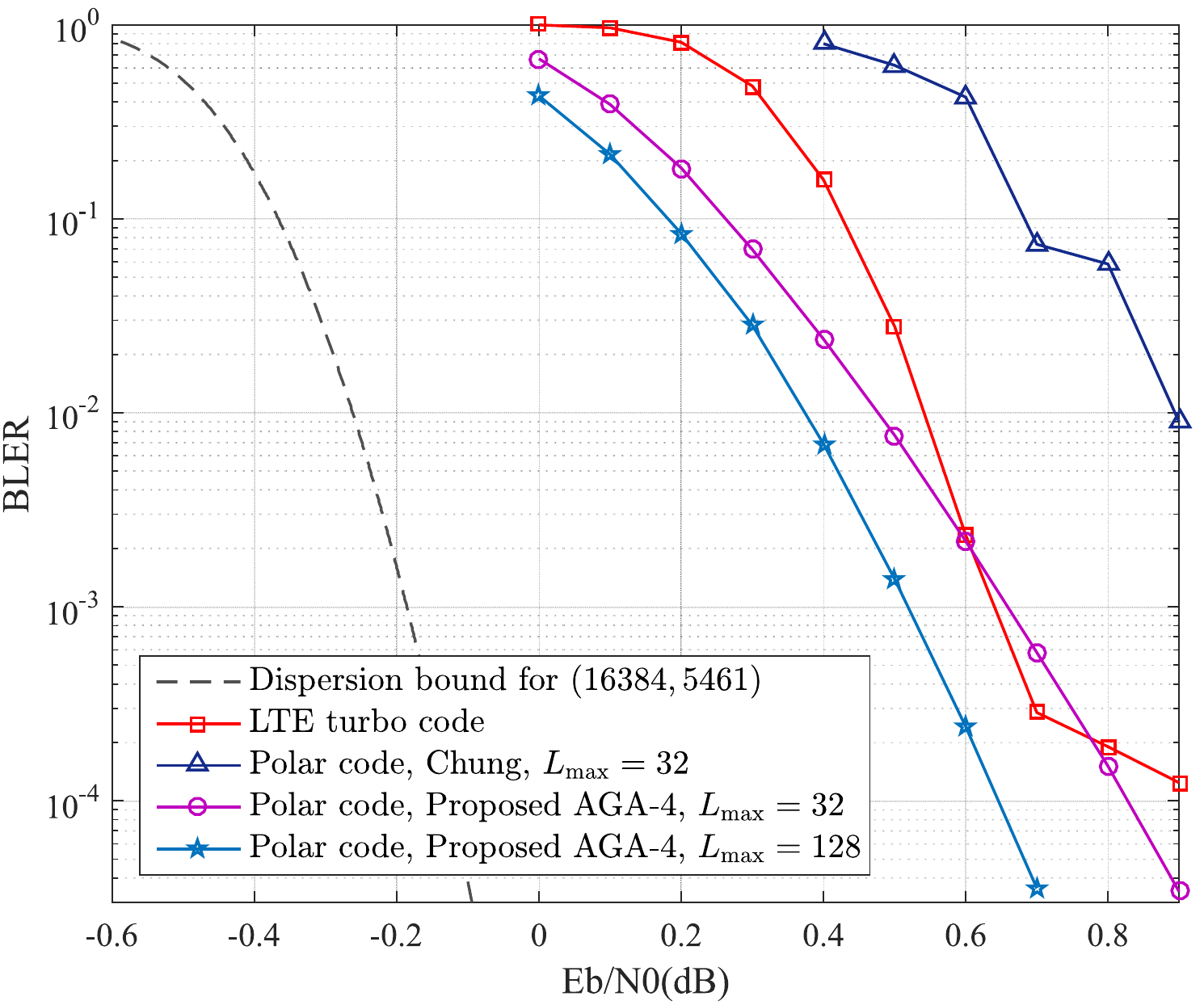}}
  \caption{The Ad-CASCL decoding BLER performance comparison of polar code and LTE turbo code with the code length $N = 2^{14} = 16384$ and the code rates $R = 1/3$. The channel is configured as the BI-AWGN.}\label{BLER_plot_3}
  \vspace{-0.7em}
\end{figure}

In Fig. \ref{BLER_plot_3}, we provide the BLER performance of polar codes by using the adaptive cyclic redundancy check (CRC) aided SC list decoding (Ad-CASCL) \cite{CAdec_niukai,aCAdec_libin}. The maximum list size in the Ad-CASCL decoder is denoted by $L_{\max}$. The 16-bit CRC in LTE standard \cite{LTE} is used. The performance of LTE turbo code is also given as a comparison, where the Log-MAP decoding is applied in the turbo decoder with 8 iterations \cite{Lin_shu_channel_coding}. We can see that Chung's conventional AGA scheme shows poor performance. It indicates that this traditional two-segment AGA scheme is not suitable for the polar code construction. On the contrary, the polar codes constructed by the proposed AGA-4 scheme perform well. When $L_{\max} = 32$, polar code can outperform LTE turbo code in the low SNR region. For the high SNR region, although turbo code sometimes performs better than polar code, it suffers from the error floor. When $L_{\max}$ is set to $128$, polar code can outperform LTE turbo code for any SNR. Additionally, this polar code constructed by AGA-4 scheme with $L_{\max} = 128$ Ad-CASCL decoding can achieve BLER $\le 10^{-3}$ at ${E_b}/{N_0} = 0.51$dB. We compare this performance to the Shannon limit at the same finite block length, which is provided in \cite{finite_shannon_limit}. The maximum rate that can be achieved at block length $N$ and BLER $\epsilon$ can be well approximated by
\begin{equation}\label{finite_shannon_limit}
  {R_{\max }} = C - \sqrt {\frac{V}{N}} {Q^{ - 1}}\left( \epsilon \right),
\end{equation}
where $C$ is the channel capacity and $V$ is a quantity called the channel dispersion that can be computed from the channel statistics, using the formula:
\begin{equation}
  V = {\mathop{\rm Var}} \left[ {\log \frac{{W\left( {Y\left| X \right.} \right)}}{{W\left( Y \right)}}} \right].
\end{equation}
For the BI-AWGN channel, the transition probability $W\left( {y\left| x \right.} \right)$ is written in (\ref{BIAWGN_trans}). The channel dispersion $V$ is written as
\begin{equation}
\begin{aligned}
  V & = \frac{1}{2}\sum\limits_{x \in {\mathbb B}} {\int_{\mathbb R} {W\left( {y\left| x \right.} \right){{\left[ {\log \left( {\frac{{2W\left( {y\left| x \right.} \right)}}{{W\left( {y\left| 0 \right.} \right) + W\left( {y\left| 1 \right.} \right)}}} \right)} \right]}^2}dy} }\\
   ~ & - {I^2}\left( W \right).
\end{aligned}
\end{equation}
By using (\ref{finite_shannon_limit}), we can calculate the Shannon limit for the $\left( {N, K} \right) = \left( {16384, 5461} \right)$ code which is named as the ``dispersion bound'' in Fig. \ref{BLER_plot_3}. To achieve a rate $R = 5461/16384 = 1/3$, the minimum ${E_b}/{N_0}$ required is $-0.186$dB. Hence, polar code constructed by AGA-4 with $L_{\max} = 128$ is $0.696$dB from the Shannon limit. When $L_{\max}$ increases, this SNR gap will become smaller.

\section{Conclusion}
In this paper, we introduced the concepts of PVS and PRS which explain the essential reason that polar codes constructed by conventional AGA expresses poor performance at long code lengths. Then we proposed a new metric, named CLE, to quantitatively evaluate the remainder error of AGA. We further derived the upper bound of CLE to simplify its calculation. Guided by PVS, PRS and CLE bound analysis, we proposed new rules to design AGA for polar codes. Simulation results show that the performance of all AGA schemes is consistent with CLE analysis. When the polarization levels increase, Chung's conventional AGA scheme suffers from a catastrophic performance jitter. On the contrary, the proposed AGA methods guided by the proposed rules stably guarantee the excellent performance of polar codes for both the AWGN channels and the Rayleigh fading channels.

\ifCLASSOPTIONcaptionsoff
  \newpage
\fi

\end{document}